\documentclass[a4paper,accepted=2023-10-24]{quantumarticle}
\pdfoutput=1

\usepackage[colorlinks]{hyperref}
\usepackage[numbers,sort&compress]{natbib}
\usepackage[phfqit]{tchshortcuts}
\usepackage{amsthm}
\usepackage[
	print-unity-mantissa=false,
	exponent-product={\cdot},
	range-phrase=--
]{siunitx}
\usepackage[compat=0.6]{yquant}
\usepackage{relsize}
\usepackage{algorithm}
\usepackage{algpseudocodex}
\usepackage[capitalise,nameinlink]{cleveref}

\newtheorem{prop}{Proposition}
\newtheorem{lem}{Lemma}
\theoremstyle{remark}
\newtheorem*{rem*}{Remark}
\crefname{step}{step}{steps}

\yquantset{/yquant/operator/minimum width=0mm}

\usepackage{xcolor}

\newcommand\qmaddress{Quantum Motion, 9 Sterling Way, London N7 9HJ, United Kingdom}
\newcommand\oxfordaddress{Department of Materials, University of Oxford, Parks Road, Oxford OX1 3PH, United Kingdom}

\begin{document}

\title{Actis: A Strictly Local Union--Find Decoder}

\author{Tim Chan}
\email{timothy.chan@materials.ox.ac.uk}
\affiliation{\oxfordaddress}
\orcid{0000-0001-6187-7402}

\author{Simon C. Benjamin}
\email{simon.benjamin@materials.ox.ac.uk}
\affiliation{\oxfordaddress}
\affiliation{\qmaddress}
\orcid{0000-0002-7766-5348}

\begin{abstract}
Fault-tolerant quantum computing requires classical hardware
to perform the decoding necessary for error correction.
The Union--Find decoder is one of the best candidates for this.
It has remarkably organic characteristics,
involving the growth and merger of data structures
through nearest-neighbour steps;
this naturally suggests the possibility of its realisation
using a lattice of simple processors
with nearest-neighbour links.
In this way
the computational load can be distributed with near-ideal parallelism.
Here we show for the first time that
this strict
(rather than partial)
locality is practical,
with a
	worst-case runtime $\mathcal O(d^3)$
	and mean runtime subquadratic in the surface code distance $d$.
A novel parity-calculation scheme is employed
which can simplify previously proposed architectures,
and our approach is optimised for circuit-level noise.
We compare our local realisation
with one augmented by long-range links;
while the latter is of course faster,
we note that local asynchronous logic
could negate the difference.
\end{abstract}

\maketitle
\section{Introduction}
\label{sec:introduction}
In the fault-tolerant era of quantum computing,
the quantum hardware must be supported by classical decoders
that infer the nature of errors `on the fly' from measurements.
It is challenging to find decoder implementations
that are sufficiently
	fast,
	compact,
	and (ideally) with low power requirements.
We focus on decoders for the \emph{surface code}
\cite{Dennis2002,Fowler2012a,Litinski2019}:
one of the most promising error-correcting codes
needed for fault tolerance
due to its simplicity.
Many decoders have been developed
such as
	\emph{minimum-weight perfect matching}
		(MWPM)
		\cite{Edmonds1965,Fowler2012},
	\emph{renormalisation group}
		\cite{Duclos-Cianci2010,Duclos-Cianci2010a},
	\emph{Markov chain Monte Carlo}
		\cite{Wootton2012},
	and various using
		\emph{belief propagation}
			\cite{Criger2018,Higgott2023b,Higgott2023a,Kuo2022},
		\emph{neural networks}
			\cite{Sheth2020,Overwater2022},
		or a \emph{hierarchical design}
			\cite{Delfosse2020a,Meinerz2022,Ravi2023_quantum_bibstyle,Smith2023}.

One simple and fast
decoder is the Union--Find decoder (UF)
designed by Delfosse et al.~\cite{Delfosse2020,Delfosse2021}.
It has a relatively high accuracy
and a mean runtime slightly higher than cubic in $d$
\cite[\S2.3]{Liyanage2023}.
Liyanage et al.~\cite{Liyanage2023} recently proposed \emph{Helios},
the FPGA
(field-programmable gate array)
implementation of
an almost-local version of UF.
By \emph{local} we mean runnable on a grid of identical nodes
(classical processors),
each communicating only with their nearest neighbours.
They report an improved mean runtime: sublinear in $d$.

\begin{figure}
	\centering
	\includegraphics[width=0.5\textwidth]{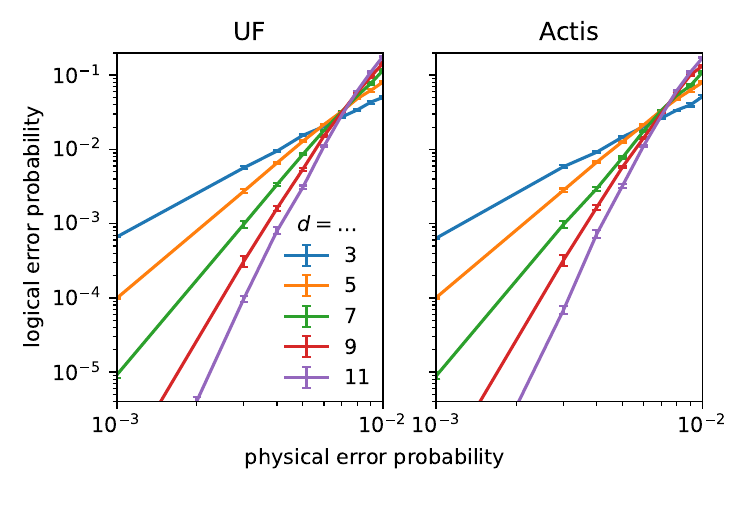}
	\caption{Threshold plots for original UF and Actis.
	We use the circuit-level noise model
	detailed in \cref{sec:circuit-level_noise_appendix}.
	These numerics are consistent with the assertion
	in the main text
	that both decoders have the same accuracy.
	Each datapoint is the mean of \numrange{1e4}{1.5e7}
	samples of decoding cycle;
	errorbars show standard error.
	The threshold is the physical error probability at which
	the lines cross,
	\num{\approx 7.5e-3} for both plots.}
	\label{fig:thresholds}
\end{figure}

\begin{figure*}
	\centering
	\includegraphics[width=\textwidth]{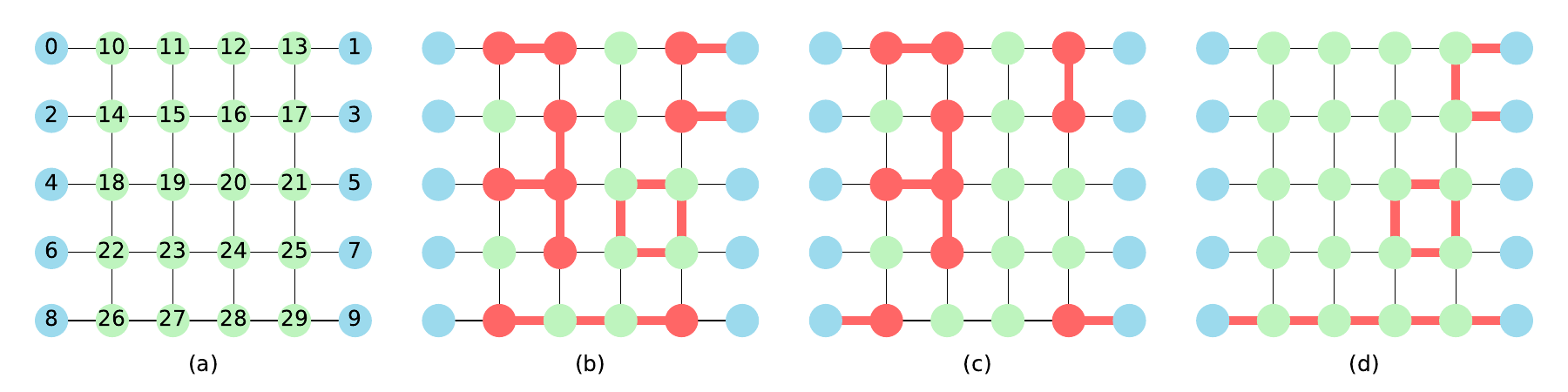}
	\caption{
	(a) The graph representing a distance-5 surface code.
	Detectors are green (nodes 10--29);
	boundary nodes, blue (nodes 0--9).
	(b) An example of an error $\mathbb E$ is the set of thick red edges
	whose corresponding syndrome $\mathbb S$ is the set of red nodes.
	Note a path in $\mathbb E$
	(like that from node 26 to 29)
	will make a defect only at each endpoint;
	a cycle will make no defects.
	The decoder sees only $\mathbb S$, not $\mathbb E$.
	(c) A correction $\mathbb C$ for $\mathbb E$ is the set of thick red edges.
	The syndrome made by $\mathbb C$
	is always the same as that of $\mathbb E$.
	(d) The leftover $\mathbb L$ after combining $\mathbb E$ and $\mathbb C$.
	This example comprises a cycle and two paths:
		one between nodes on the same boundary
		and one between opposite boundaries.
	The latter represents a logical bitflip.}
	\label{fig:surface_code}
\end{figure*}

Our main results
extend this literature,
especially the ideas proposed for Helios,
as we introduce a paradigm which
reduces
	the memory requirements of each node,
	the number and size of messages passed around,
	the grid architecture complexity,
	and the total number of algorithm stages.
This enhancement,
used in our version of almost-local UF called \emph{Macar},
is a simple scheme inspired by anyon annihilation
\cite{Kitaev2003}
to calculate parities.
We also design \emph{Actis}:
a \emph{strictly} local version of UF,
the first of its kind,
which is even more practical
due to the lack of long-range links,
and whose mean runtime is subquadratic in $d$
for error probabilities below threshold.
While Helios was built only for phenomenological noise,
we design our algorithms for circuit-level noise.
Actis further takes advantage of the pre-existing decoder structure
for this noise to minimise its physical overhead.
Having compared Macar and Actis
we note that the use of a fast communication relay
implemented with asynchronous logic
constitutes a third version of UF
which rivals the speed of Macar
whilst maintaining strict locality.

As well as having improved mean and worst-case runtime scalings,
all three of our versions
are exact implementations of,
hence just as accurate as,
original UF.
We confirm the latter in \cref{fig:thresholds},
noting our versions behave identically in their output.

The paper proceeds as follows.
In \cref{sec:background}
we cover the prerequisite background theory.
In \cref{sec:almost_local_uf}
we discuss almost-local UF as introduced in the Helios scheme
and our streamlined design, Macar,
before describing our strictly local Actis
in \cref{sec:actis}.
We evaluate runtime for Macar and Actis
in \cref{sec:runtime_analysis}
and discuss how Actis can be sped up using asynchronous logic
in \cref{sec:asynchronous_logic}.
\Cref{sec:conclusion} concludes.

Our emulation code is on GitHub at \cite{Chan2023a_quantum_bibstyle}
and the raw data for the plots in this paper is available at \cite{Chan2023b_quantum_bibstyle}.

\section{Background}
\label{sec:background}
\Cref{sec:surface_code_and_decoding} covers the surface code
with a focus on decoding.
We approach this standard material using graph theory
as it facilitates a natural explanation of UF
presented in \cref{sec:uf}.

\subsection{Surface Code and Decoding}
\label{sec:surface_code_and_decoding}
As its name suggests,
the surface code arranges the physical qubits on a 2D grid.
It encodes one logical qubit and corrects for both bit- and phaseflip errors
(the ability to correct these two discrete errors
allows us to correct any qubit error \cite[\S 10.3.1]{Nielsen2010}).
Here, we describe how to correct bitflip errors;
phaseflip errors are corrected analogously.
We explain the decoding cycle
under three noise models of increasing realism:
\emph{code capacity},
used in \cref{sec:code_capacity},
assumes perfect measurements;
\emph{phenomenological noise},
described in in \cref{sec:phenomenological_noise},
generalises to faulty measurements.
\Cref{sec:circuit-level_noise} discusses
\emph{circuit-level depolarising noise}.

\subsubsection{Code Capacity}
\label{sec:code_capacity}
\begin{figure*}
	\centering
	\includegraphics[width=\textwidth]{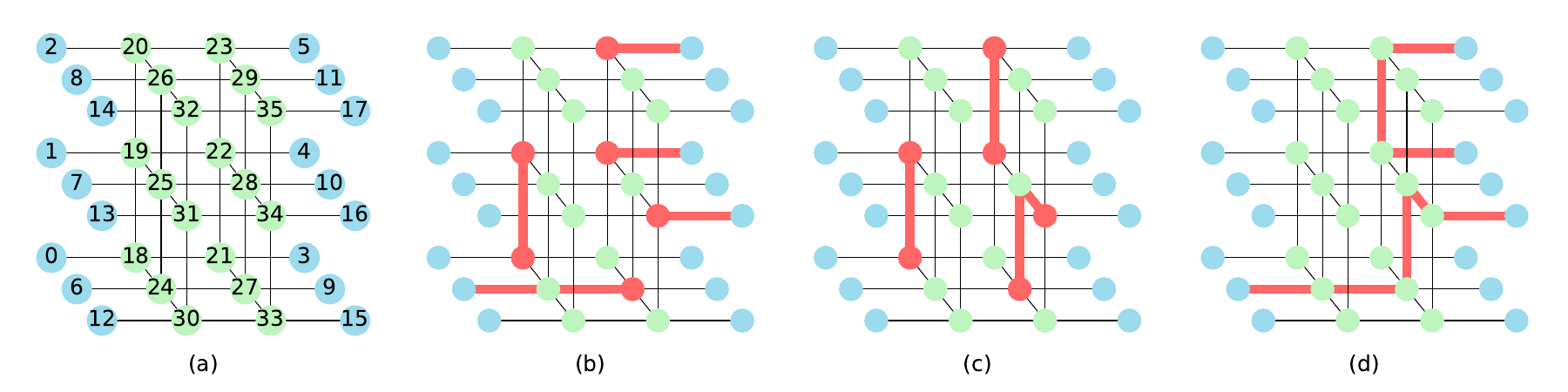}
	\caption{The 3D generalisation of \cref{fig:surface_code}
	to account for faulty measurements.
	(a) Instead of distance-5,
	we show a distance-3 code
	for visual clarity.
	The time axis points upward.
	(b) The error $\mathbb E$ can now include vertical `timelike' edges
	like that from node 18 to 19.
	The syndrome $\mathbb S$ is the set of red nodes
	throughout the whole lattice.
	(c) A correction $\mathbb C$ for $\mathbb E$.
	(d) The leftover $\mathbb L$ comprises
		one path between the same boundary
		and one between opposite boundaries.
	The latter represents a logical bitflip.}
	\label{fig:3D_surface_code}
\end{figure*}

The surface code is based on the graph $G =(V, E)$
in \cref{fig:surface_code}(a)
which comprises a set $V =V_\d \cup V_\b$ of nodes
arranged on a $d \times (d+1)$ grid
and a set $E$ of edges.
Members of $V_\d$ are called detectors;
$V_\b$, boundary nodes.
Each edge corresponds to a set of two nodes from $V$.
Physically,
	each detector represents an ancilla qubit used for measurement
	while each edge represents a data qubit.
Boundary nodes represent nothing physical
but exist just to allow every data qubit to be represented by an edge.
We will treat the decoding problem abstractly in terms of nodes and edges.
In a decoding cycle:
\begin{enumerate}
	\item \textbf{Noise corrupts our system}
	\label[step]{item:noise_corrupts_our_system}
	Each edge bitflips with some probability $p$
	called the \emph{physical error probability}
	i.e.\ is assigned bit value 1 with probability $p$,
	and bit value 0 otherwise.
	The \emph{error} $\mathbb E \subseteq E$ is the set of bitflipped edges. 
	\item \textbf{Ancilla qubits flag the noise}
	\label[step]{item:ancilla_qubits_flag_the_noise}
	Every detector records
	the parity of the edges in $\mathbb E$ which are incident to it.
	If it records odd parity, it is a \emph{defect}.
	The \emph{syndrome} $\mathbb S \subseteq V_\d$
	is the set of defects.
	\Cref{fig:surface_code}(b) shows an error and its syndrome.
	Mathematically, $\mathbb S =\sigma(\mathbb E)$ where
	\begin{equation}\label{eq:syndrome_definition}
		\sigma(\mathbb E)
		:=V_\d \cap \bigsd_{e \in \mathbb E} e
	\end{equation}
	and $\bigsd_{e \in \mathbb E} e :=f \sd \dots \sd g$
	for $\mathbb E =`{f, \dots, g}$
	where $\sd$ denotes symmetric difference.
	\item \textbf{The classical decoder acts}
	Given $\mathbb S$,
	the decoder must output a \emph{correction} $\mathbb C \subseteq E$
	which, if treated as an error, would make the same syndrome $\mathbb S$.
	In other words, $\mathbb C$ must satisfy $\sigma(\mathbb C) =\mathbb S$.
	\Cref{fig:surface_code}(c) shows a correction.
	\item \textbf{Success or failure}
	\label[step]{item:success_or_failure}
	The \emph{leftover} $\mathbb{L :=E \sd C}$
	is then guaranteed to satisfy the following lemma
	(for completeness, a proof is given in
	\cref{sec:proof_of_lem:leftover}).
	\begin{lem}\label{lem:leftover}
	$\mathbb L$ comprises only cycles
	or paths between distinct boundary nodes.
	\end{lem}
	The question of whether the decoder has been successful
	depends on whether a logical error has resulted from the correction.
	Any path between opposite boundaries represents a logical bitflip.
	Hence, if $\mathbb L$ has an odd number of such paths
	the logical qubit has picked up a logical error.
	\Cref{fig:surface_code}(d) shows a leftover with one such path.
\end{enumerate}
The aim of the decoder is to output a correction
which is least likely to lead to a logical error.
\cref{fig:thresholds} shows examples of how,
for a surface code decoder,
logical error probability depends on both
	$d$ and $p$.
Each decoder has a different \emph{threshold}
i.e.\ physical error probability $p_\th$
below (above) which
increasing $d$
decreases (increases) the logical error probability.
The further below threshold a decoder operates,
the more effective it is.
In the case of UF,
operating at its threshold is practically useless;
significant benefits of error correction are reaped
only when $p$ is several times below threshold
i.e.\ $p \ll p_\th$.

\begin{figure*}
	\centering
	\includegraphics[width=\textwidth]{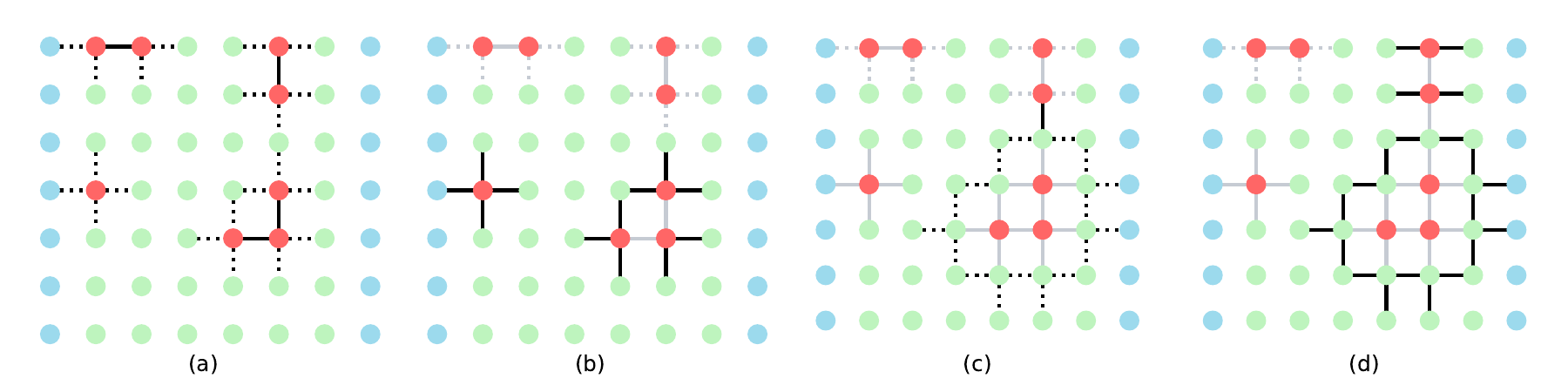}
	\caption{A syndrome validation example which needs four growth rounds.
	Ungrown edges are invisible;
	half-grown, dotted;
	fully grown, solid.
	Edges which have grown since the last round are black;
	else, grey.
	(a) Each defect grows its incident edges by $\frac12$.
	Note some edges have grown by $\frac12$ twice, so are fully grown.
	(b) The two active clusters grow again.
	The west one touches a boundary hence inactivates.
	(c) The remaining active cluster grows
	and touches hence merges with the inactive cluster north of it.
	(d) The resultant active cluster
		grows
		and touches the boundary hence inactivates.}
	\label{fig:syndrome_validation}
\end{figure*}
\subsubsection{Phenomenological Noise}
\label{sec:phenomenological_noise}
In reality,
ancilla qubit measurement is faulty
i.e.\ in \cref{item:ancilla_qubits_flag_the_noise}
each detector records the \emph{wrong} parity with some probability $q$.
The result is that the syndrome is subject to some noise:
$\mathbb S =\sigma(\mathbb E) \sd \mathbb F$,
where $\mathbb F \subseteq V_\d$
is the set of nodes which record the wrong parity for that measurement round.
Note boundary nodes cannot be in $\mathbb F$
as they do not represent ancilla qubits which are measured.

The usual way to account for faulty measurements is
to measure each ancilla qubit
	not once
	but $\tau+1$ times per decoding cycle,
where $\tau >0$
(intuitively,
the higher $q$ is,
the higher we should make $\tau$).
We then have
	not one 2D sheet of syndrome data as in \cref{fig:surface_code}(b),
	but $\tau+1$ of them,
	which we label $`{\mathbb S_0, \dots, \mathbb S_\tau}$.
In general there is a new $\mathbb E$ and $\mathbb F$ for each sheet,
but the data is cumulative in $\mathbb E$ i.e.
\begin{equation}\label{eq:2D_sheet}
\mathbb S_t =`\big[\bigsd_{u=0}^t \sigma(\mathbb E_u)] \sd \mathbb F_t
\qquad \forall t \in 0..\tau
\end{equation}
so to localise each $\mathbb E_t$
we take the difference between consecutive sheets:
\begin{align}
\upDelta \mathbb S_t &:= \mathbb S_t \sd \mathbb S_{t-1}
\qquad \forall t \in 1..\tau \notag\\
&\overset{\eqref{eq:2D_sheet}}=
\sigma(\mathbb E_t) \sd \mathbb F_t \sd \mathbb F_{t-1}.
\label{eq:2D_difference_sheet}
\end{align}
We then stack these 2D `difference sheets'
$`{\upDelta \mathbb S_1, \dots, \upDelta \mathbb S_\tau}$ to make
a 3D simple cubic lattice as in \cref{fig:3D_surface_code}(a).
Sheets are stacked such that
$\upDelta \mathbb S_1$ is on the bottom,
$\upDelta \mathbb S_2$ is the next one above,
etc.
This way,
sheets which derive from later measurements are higher up the stack
so we can interpret the upward direction as going forward in time
(hence the subscript $t$).

Each detector no longer represents one ancilla qubit;
rather,
the difference between two consecutive measurements of that ancilla qubit
at a given point in time.
Each horizontal edge no longer represents one data qubit
but a possible time at which that data qubit could bitflip.
Vertical edges are a new addition:
each one represents one possible faulty measurement.
Specifically,
each edge between
$\upDelta \mathbb S_t$ and
$\upDelta \mathbb S_{t+1}$
corresponds to a possible member of $\mathbb F_t$.
This makes sense
e.g.\ a member of
$\mathbb F_1$ effects a change in both
$\upDelta \mathbb S_1$ and
$\upDelta \mathbb S_2$
as per \cref{eq:2D_difference_sheet}.

If we redefine $G$ as this 3D simple cubic lattice,
all steps of the decoding cycle are the same as in
\cref{sec:code_capacity}
with the modification in \cref{item:noise_corrupts_our_system} that
	each \emph{horizontal} edge bitflips with probability $p$;
	each \emph{vertical} edge, $q$.
Detectors still always record the \emph{correct} parity
as now faulty measurements are modelled by vertical edges.
\Cref{fig:3D_surface_code}(b--d) shows a decoding cycle example.

In our emulations we assume the \emph{batch decoding scheme}
i.e.\ we set $\mathbb F_0 =\mathbb F_\tau =\varnothing$,
which explains the absence of vertical edges
	below $\upDelta \mathbb S_1$ and
	above $\upDelta \mathbb S_\tau$
in \cref{fig:3D_surface_code}.
This simplification allows us to
determine the success of the decoder
for a batch of $\tau+1$ measurement rounds in isolation.
In contrast,
continuously decoding a constant stream of syndrome data
from an indefinite number of measurement rounds
is the topic of
\emph{stream decoding}.
We do not treat this in our paper
as methods already exist
allowing most batch decoders,
including our versions of UF,
to be used as stream decoders
\cite{Dennis2002,Tan2022,Skoric2023}.

\subsubsection{Circuit-Level Noise}
\label{sec:circuit-level_noise}
Yet a more realistic noise model is
the circuit-level depolarising model,
which considers the quantum circuit
performed during a measurement round
and a variety of possible faults it can make.
Details of this are provided in \cref{sec:circuit-level_noise_appendix}
but the end result is the addition of certain diagonal edges in $G$,
as \cref{fig:SS_tree}(a) shows.
The bitflip probability for each edge
is a function of some characteristic
physical error probability $p$;
this function potentially varies for each edge.

Circuit-level noise depends on
parameters not standardised in literature
so it is difficult to compare thresholds under this model.
This is not the case for
	code capacity
	and phenomenological noise
	(standardised by setting $q =p$ and $\tau =d$),
under which we numerically obtain thresholds for UF of
	\num{\approx 9.8e-2}
	and \num{\approx 2.6e-2}
respectively,
consistent with previous results \cite[\S5.4]{Hu2020a}.

MWPM is a popular decoder which
	finds a most likely error for the syndrome,
	and outputs that as the correction
	i.e.\ $\mathbb C =\arg\max`{\pr(\mathbb E): \sigma(\mathbb E) =\mathbb S}$.
For small $p$ under code capacity or standardised phenomenological noise,
this amounts to finding a correction whose size $|\mathbb C|$ is minimal.
While accurate,
MWPM has a worst-case runtime
$\mathcal O(N^3 \lg N)$ \cite[\S 3]{Higgott2022a}
where $N :=|V| =\mathcal O(d^2)$ for code capacity
and $\mathcal O(d^3)$ for the other noise models.
UF approximates MWPM \cite{Wu2022}
with an improved worst-case runtime $\mathcal O`\big(N \alpha(N))$
where $\alpha(N) \le 3$ for all practical $d$.
In \cref{sec:comparison_with_mwpm}
we compare the accuracy of MWPM with UF
for practical noise levels
and discuss their near-term relevance.

For visual clarity,
\crefrange{fig:syndrome_validation}{fig:almost_local_burning_peeling}
use $G$ for code capacity.
However,
the reader should keep in mind
the rest of the paper
(save \cref{sec:actis}
which focuses on circuit-level noise)
also applies to the other noise models.

\subsection{UF}
\label{sec:uf}
Conceptually, UF groups defects into spatial clusters
then finds a correction only from edges within a cluster.
UF keeps the total size of all clusters small
to ensure the correction is small.
The algorithm comprises the following stages:
\begin{enumerate}
	\item \textbf{Syndrome Validation}
	An active cluster nucleates at each defect
	which grows outward in all directions at the same speed.
	\item \textbf{Spanning Forest Growth}
	A spanning tree is grown within each cluster.
	Using our parity-calculation scheme,
	this stage comes for free.
	\item \textbf{Peeling} Each edge of each tree is peeled
	i.e.\ removed from the tree
	and included or excluded from the correction
	depending on the syndrome.
\end{enumerate}

\subsubsection{Syndrome Validation}
\Cref{fig:syndrome_validation} shows a syndrome validation example.
Mathematically, a cluster $C =(V_C, E_C)$ is a connected subgraph of $G$.
A cluster is \emph{active} iff there exists no correction for its defects
using only edges in the cluster i.e.
$\nexists \mathbb C \subseteq E_C :\sigma(\mathbb C) =\mathbb S \cap V_C$.
Else it is \emph{inactive}.
Intuitively,
an inactive cluster is one where we have accounted for all defects
with a correction local to that cluster;
an active cluster still seeks such a correction.
In practice we determine activity by the following lemma
adapted from \cite[Lemma~1]{Delfosse2021}
and proved in \cref{sec:proof_of_lem:cluster_activity}.
\begin{lem}\label{lem:cluster_activity}
Cluster $C$ is active iff it has an odd defect count
and touches no boundary i.e.
\begin{equation}
`\big(\textnormal{$|\mathbb S \cap V_C|$ odd})
\wedge
`\big(V_\textnormal{boundary} \cap V_C =\varnothing).
\end{equation}
\end{lem}
\noindent
Before syndrome validation,
a cluster $(`{v}, \varnothing)$ is initialised for every $v \in V$
(so there is an active cluster for each defect
and an inactive one for each nondefect).
Each edge in $E$ has a \emph{growth value} in $`{0, \frac12, 1}$
which is initialised at 0.
Syndrome validation proceeds via some number of growth rounds
and stops when all clusters are inactive.
In a growth round:
\begin{enumerate}
	\item Each active cluster grows all edges around it by $\frac12$.
	\item Each new fully grown edge
	merges the two clusters at its endpoints into one bigger cluster,
	if the endpoints did not already belong to the same cluster.
	Mathematically, if clusters $A$ and $B$ merge along edge $e$
	the resultant cluster is
	\begin{equation}
	C =`\big(V_A \cup V_B, E_A \cup E_B \cup `{e}).
	\end{equation}
\end{enumerate}
One may wonder why edges grow by only $\frac12$ per growth round.
If they grew by 1 per round,
a pair of defects separated by two edges would merge
just as fast as a pair separated by one edge.
Growing edges by $\frac12$ per round distinguishes these cases,
reducing the number of fully grown edges
hence the size of all clusters.

Syndrome validation has a runtime $\mathcal O`\big(N\alpha(N))$
if the Union--Find data structure
(after which the Union--Find decoder is named)
is used \cite{Tarjan1975}.
It ensures the \emph{existence} of a correction using only cluster edges
while the next two stages,
shown in \cref{fig:growth_and_peeling},
find an \emph{instance} of one.
\begin{figure}
	\centering
	\includegraphics[width=0.5\textwidth]{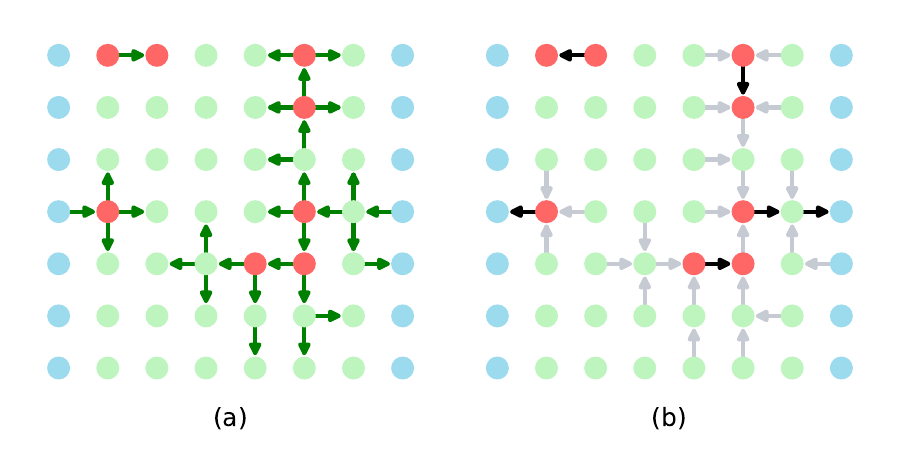}
	\caption{Continuing the example in \cref{fig:syndrome_validation}:
	(a) A spanning forest is grown in the direction of the arrows.
	Each arrow is an edge in the forest
	and points from parent to child node.
	(b) The forest is peeled in the direction of the arrows.
	The set of black arrows is the correction returned by the decoder.}
	\label{fig:growth_and_peeling}
\end{figure}

\subsubsection{Spanning Forest Growth}
For each cluster $C$ a root node is chosen as follows:
if $C$ touches a boundary, one of its boundary nodes is chosen;
else, any node is chosen.
A spanning tree of $C$
(a subgraph of $C$ that is
connected,
acyclic,
and spans $V_C$)
is then grown from this root.
The union of all the trees is the forest.
This stage is done by breadth- or depth-first search;
either way, the runtime is linear
in the number of edges in the forest i.e.\ $\mathcal O(N)$.

\subsubsection{Peeling}\label{sec:peeling}
A correction $\mathbb C =\varnothing$ is initialised
then the forest is traversed edge by edge, starting from the leaves.
Each time a defect is encountered all edges traversed thereafter,
until another defect is encountered,
are added to $\mathbb C$.
Once the whole forest has been traversed
$\mathbb C$ will satisfy $\sigma(\mathbb C) =\mathbb S$
\cite[Theorem~1]{Delfosse2020}.
Since each edge in the forest is traversed exactly once
the runtime is $\mathcal O(N)$.

\subsubsection{Weighted Edges}
Under phenomenological and circuit-level noise,
the bitflip probability may vary for each edge.
In the case these probabilities are known,
UF can be made more accurate by
modifying syndrome validation to grow clusters along weighted edges
i.e.\ whose growth values are more granular than
$`{0, \frac12, 1}$ \cite{Huang2020}.

This concludes our review of
	the well-established surface code
	and UF
paradigms,
using the notation that we require for the remainder of the paper.
We now briefly discuss
recent literature working towards local architectures for implementing UF,
especially Helios.

\section{Almost-Local UF}\label{sec:almost_local_uf}
It is clear the nature of UF is quite local:
	all edges involved connect neighbouring nodes,
	and two clusters communicate only when they touch;
so it seems sensible to implement UF as locally as possible.
This may have advantages over the standard picture
in which the decoder protocol is implemented using a
	single,
	fully fledged
classical processor,
located remotely from the actual qubits.
Instead,
one can imagine a more distributed decoder comprising
	a very simple classical processor
		corresponding to each node
	and an interprocessor communication link
		corresponding to each edge.
We can visualise this local architecture
as directly corresponding to \cref{fig:SS_tree}(a) if
	nodes represent processors and
	edges represent the links.

This alternative paradigm has various benefits.
	Firstly, its computation is parallelised:
		non-neighbouring computations can be run
		independently and simultaneously;
		we will see this leads to superior runtime scaling.
	Secondly, it allows
		each processor to sit closer to its data source (the ancilla qubit),
		simplifying wiring and reducing signal losses
		\cite[p \numrange{6}{7}]{Vandersypen2017}.
	Thirdly, there may be implications in terms of robustness:
		if one part breaks down then only that part,
		rather than the whole decoder,
		must be replaced.

In \cref{sec:helios} we give a conceptual overview of
Helios developed by Liyanage et al.~\cite{Liyanage2023} and in
\cref{sec:macar} we describe Macar in detail,
including our parity-calculation scheme
that simplifies the physical architecture
and allows a spanning forest to be grown in one timestep regardless of $d$.

\subsection{Helios}\label{sec:helios}
The authors consider only syndrome validation
as this is the only stage that is nontrivial to implement locally.
Their implementation uses the local architecture just described
i.e.\ there is a
	processor for each node in $G$ and a
	communication link for each edge.
For the rest of this section we refer to
processors as nodes and
communication links as edges.

Clusters are identified by a unique integer stored by all its nodes.
When two clusters merge
the higher-integer cluster becomes part of the other one.
This occurs via a flooding algorithm:
each node looks at its neighbours along fully grown edges;
if it sees a lower integer than its own,
it stores that lower integer instead.
This results in a wave of change
propagating across the higher-integer cluster,
starting from where the two clusters touched.
After the flood,
all the nodes of the higher-integer cluster will store the lower integer.
Generalising to three or more clusters merging at once is simple:
a flood propagates from each touching point but eventually
only the lowest-integer cluster will remain.

As well as merging,
the other task in syndrome validation is
determining the activity of each cluster.
The authors devised the following messaging system to tackle this.
Each cluster always has a well-defined root node.
Whenever a defect changes its stored integer,
it sends a message to the root of the cluster it has just joined.
Every root stores a bit and flips it whenever it receives a message.
Once all flooding has finished
and all messages have reached their destination,
the bit of each root will equal
its cluster's defect count parity.
This,
by \cref{lem:cluster_activity},
determines the activity of clusters which do not touch a boundary.
The authors do not mention treatment of clusters which touch a boundary
but we do so in \cref{sec:syndrome_validation}.
These messages are sent via a network
comprising
	the edges of $G$
	and additional edges between non-neighbouring nodes
but our modification removes the need for them.

In addition to the nodes,
there is a central controller
which manages the global variables of the UF algorithm.
It connects to each node
via a hierarchical tree
which allows restricted communication:
the controller can broadcast the same message to all nodes
and receive Boolean messages from any node,
but does not know the sender.
This is why we say Helios is \emph{almost} local.
The authors leave tree height as a design choice;
in this paper we follow their implementation,
suitable for low code distances,
and set it to 1
i.e.\ connect the controller directly to each node.
This allows us to
ignore the time cost
of controller--node communication.

In this limit,
the worst-case runtime of Helios is $\mathcal O(N)$
as the maximum possible cluster diameter
(defined for a graph as the maximum length of a shortest path)
is $\mathcal O(N)$
and the duration to flood such a cluster
(example in \cref{fig:snake})
\begin{figure}
	\centering
	\includegraphics[width=0.25\textwidth]{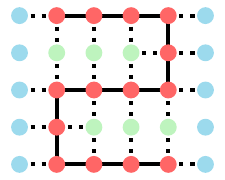}
	\caption{A maximum-diameter cluster
	is a winding snake-like path of length $\mathcal O(N)$
	where $N$ is the number of nodes in the graph $G$.}
	\label{fig:snake}
\end{figure}
scales linearly with its diameter.
This is formally better
than the $\mathcal O`\big(N \alpha(N))$ of original UF
owing to the parallelism of Helios.
More impressive is the \emph{mean} runtime which,
for phenomenological noise
i.e.\ $N =\mathcal O(d^3)$,
the authors report as sublinear in $d$
from numerics.

\subsection{Macar}\label{sec:macar}
We assume decoding to be cellular automata-like
i.e.\ via discrete timesteps with
	each component performing one operation,
	and messages travelling one edge,
per timestep.

The controller has a variable attribute
\verb|stage| indicating which stage it is in.
It is initialised to \verb|growing|
and follows the flowchart in \cref{fig:stage_flowchart}.
\begin{figure}
	\centering
	\includegraphics[width=0.4\textwidth]{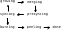}
	\caption{Flowchart for the \texttt{stage} variable
	which alternates between fixed- and arbitrary-duration stages.}
	\label{fig:stage_flowchart}
\end{figure}
There are seven stages:
	\verb|growing|,
	\verb|merging|,
	\verb|presyncing|,
	\verb|syncing|,
	\verb|burning|,
	\verb|peeling|,
	\verb|done|.
The odd stages last exactly one timestep;
the even stages, an arbitrary number of timesteps.
During syndrome validation
the controller cycles through the first four stages
once per growth round.
\verb|burning| is for spanning forest growth,
\verb|peeling| is for peeling,
and \verb|done| indicates the end of the decoding cycle.

\begin{figure*}[t]
\begin{minipage}{\linewidth}
\begin{algorithm}[H]
\caption{Run by controller every timestep.}
\label{alg:controller_advance}
\begin{algorithmic}
	\Procedure{AdvanceMacarController}{}
		\If{$(\text{not}~v.\texttt{busy})~\forall v \in V$}
			\If{$\texttt{stage} = \texttt{syncing}$}
				\If{$(\text{not}~v.\texttt{active})~\forall v \in V$}
					\State $\texttt{stage} \gets \texttt{burning}$
				\Else
					\State $\texttt{stage} \gets \texttt{growing}$
				\EndIf
			\Else
				\State $\texttt{stage} \gets$ next stage
				\Comment{Stage order is as introduced in \cref{sec:macar}.}
			\EndIf
		\EndIf
	\EndProcedure
\end{algorithmic}
\end{algorithm}
\end{minipage}
\end{figure*}

\begin{figure*}
\begin{minipage}{\linewidth}
\begin{algorithm}[H]
\caption{The four procedures each node can run during syndrome validation.}
\label{alg:node_advance}
\begin{algorithmic}
	\Procedure{Growing}{$v$}
		\If{$v.\texttt{active}$}
			\ForAll{unfully grown edges $e$ incident to $v$}
				\State $e.\texttt{growth} \gets e.\texttt{growth} +\frac12$
			\EndFor
		\EndIf
	\EndProcedure
	\State
	\Procedure{Merging}{$v$}
		\State $v.\texttt{busy} \gets \texttt{false}$
		\If{($v \in V_\d$) and ($v$ not a root) and ($v.\texttt{anyon}$)}
			\State $v.\texttt{busy} \gets \texttt{true}$
			\State relay anyon in direction of $v.\texttt{pointer}$
		\EndIf
		\For{$u \in \operatorname{access} v$}
			\If{$u.\texttt{CID} < v.\texttt{CID}$}
				\State $v.\texttt{busy} \gets \texttt{true}$
				\State $v.\texttt{pointer} \gets$ toward $u$
				\State $v.\texttt{CID} \gets u.\texttt{CID}$
			\EndIf
		\EndFor
	\EndProcedure
	\State
	\Procedure{Presyncing}{$v$}
		\State $v.\texttt{active} \gets$
		($v \in V_\d$) and ($v$ a root) and ($v.\texttt{anyon}$)
	\EndProcedure
	\State
	\Procedure{Syncing}{$v$}
		\State $v.\texttt{busy} \gets$
		(not $v.\texttt{active}$) and
		($\exists u \in \operatorname{access} v: u.\texttt{active}$)
		\State $v.\texttt{active} \gets v.\texttt{active}$ or $v.\texttt{busy}$
	\EndProcedure
\end{algorithmic}
\end{algorithm}
\end{minipage}
\end{figure*}

Each node $v \in V$ is assigned a unique integer \verb|ID|
as in \cref{fig:SS_tree}(a)
and has the following variable attributes:
\begin{itemize}
	\item \verb|defect| is a boolean indicating whether $v \in \mathbb S$.
	\item \verb|active| is a boolean indicating
	whether the cluster that $v$ is in is active.
	Since $v$ starts in its own cluster,
	$v.\verb|active|$ is initialised as $v.\verb|defect|$.
	If $v.\verb|active| = \verb|true|$ we say $v$ is active.
	\item \verb|CID| is an integer indicating
	the cluster that $v$ belongs to.
	Clusters are identified by the lowest \verb|ID| of all its nodes.
	The node of this \verb|ID| is the \emph{root} of the cluster.
	Since $v$ starts in its own cluster,
	$v.\verb|CID|$ is initialised as $v.\verb|ID|$.
	\item \verb|anyon| is a boolean indicating whether $v$ has an anyon.
	This anyon can be thought of as a particle which is passed between nodes
	during syndrome validation.
	When two anyons meet
	i.e.\ a node receives two anyons in a timestep,
	they annihilate.
	$v.\verb|anyon|$ is initialised as $v.\verb|defect|$.
	The idea is to accumulate at the root all the anyons in a cluster
	so they annihilate on the way or on arrival.
	Eventually zero or one anyon will remain at the root,
	indicating the cluster's defect count parity.
	\item \verb|pointer| is the direction $v$ should relay the anyon.
	Possible values (under circuit-level noise) are
		\texttt{C},
		\texttt{N},
		\texttt{W},
		\texttt{E},
		\texttt{S},
		\texttt{D},
		\texttt{U},
		\texttt{NU},
		\texttt{WD},
		\texttt{EU},
		\texttt{SD},
		\texttt{NWD},
		\texttt{SEU},
	representing respectively
		centre,
		north,
		west,
		east,
		south,
		down,
		up,
	and combinations thereof.
	Following the pointers starting from $v$ should eventually lead to
	the root of the cluster that $v$ is in.
	Only roots have $\verb|pointer| = \verb|C|$
	as they do not relay anyons but accumulate them.
	$v.\verb|pointer|$ is initialised to \verb|C|.
	Together, \verb|anyon| and \verb|pointer| replace
	Helios' messaging system.
	\item \verb|busy| is a boolean used to indicate whether $v$ is busy.
	$v.\verb|busy|$ is initialised to \verb|false|.
\end{itemize}

\begin{figure*}
	\centering
	\includegraphics[width=0.8\textwidth]{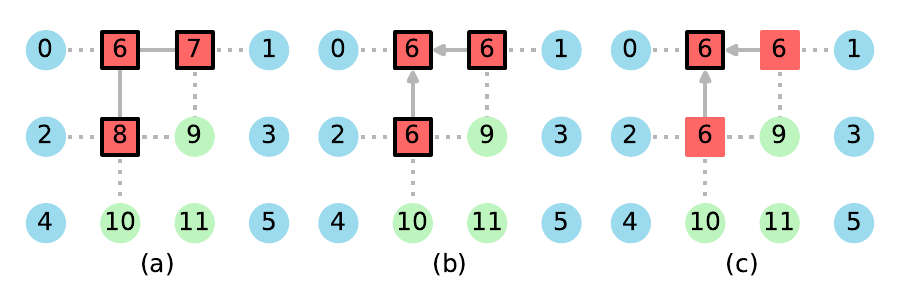}
	\caption{An example of a \texttt{merging} stage which needs three timesteps.
	Active nodes are square-shaped;
	inactive nodes are circular.
	\texttt{CID} is shown as a label.
	Nodes with anyons are outlined in black.
	Pointers are shown by arrows on edges.
	(a) The first \texttt{growth} stage has just finished.
	(b) Nodes 7 and 8 change their \texttt{CID} to 6 and point to node 6
	which is now the root of a three-node cluster.
	(c) Nodes 7 and 8 relay their anyons to the root;
	the two anyons annihilate and one remains,
	indicating the cluster's defect count is odd.
	This cluster touches no boundary so by \cref{lem:cluster_activity}
	is active.}
	\label{fig:almost_local_merging}
\end{figure*}

Each edge has a variable attribute \verb|growth| equal to its growth value,
initialised as $0~\forall e \in E$.
We restrict growth values to $`{0, \frac12, 1}$
and leave granularising for future work,
but allow UF under circuit-level noise
to grow along the diagonal edges
as we found this gave a $1.5 \times$ improvement to the threshold.
Since processing is done at nodes rather than edges
we can for each edge pick one of its endpoints,
say the one of the lower \verb|ID|,
to be responsible for storing its \verb|growth|.
The set of edges a node is responsible for,
is then the \emph{owned edges}
\begin{equation}
v.\texttt{owned} :=`{uv: uv \in E~\text{and}~v.\texttt{ID} < u.\texttt{ID}}
\end{equation}
stored as a constant attribute.

In each timestep the controller runs \cref{alg:controller_advance},
which simply moves \verb|stage| on if no node is busy.
This is where the controller--node connectivity is used:
the controller need not know \emph{which} node is busy;
rather, if \emph{a} node is busy.
After every \verb|syncing| stage
the controller also checks for any active clusters.
If there are,
	another growth round is needed
	so \verb|stage| is reset to \verb|growing|.
If none,
	syndrome validation is done so \verb|stage| is set to \verb|burning|.

Once \verb|stage| reaches \verb|done|
the controller stops running \cref{alg:controller_advance}
and the decoder outputs the correction $\mathbb C$ made during peeling.
In this paper we do not consider how the decoder
	offloads $\mathbb C$
	and loads in the next $\mathbb S$
as these are topics more relevant to stream decoding.

\subsubsection{Syndrome Validation}
\label{sec:syndrome_validation}
It is helpful to define the \emph{access} of a node
as the set of neighbours along fully grown edges:
\begin{equation}
\operatorname{access} v
:=`{u \in V: uv \in E \wedge uv.\texttt{growth} = 1}.
\end{equation}
In a timestep
each node runs whichever procedure matches \verb|stage|
out of the four in \cref{alg:node_advance}.

Stage \verb|growing| is simply a change of \verb|growth|s:
each node in an active cluster
increments the \verb|growth| of edges around it,
if not already fully grown.

\Cref{fig:almost_local_merging} shows a \verb|merging| example.
In \verb|merging| two processes occur:
flooding and anyon annihilation.
Flooding is as described in \cref{sec:helios}:
each node looks at its access;
if it sees a lower \verb|CID| than its own,
it updates its own to match it.
Additionally,
the node updates its \verb|pointer| toward that neighbour.
This ensures following the pointers
always leads to the correct root promised by \verb|CID|.
In anyon annihilation,
any node with an anyon relays it in the direction of its pointer.
At the end of this stage,
a cluster's defect count is odd iff its root has an anyon.

Ideally,
pointers are updated \emph{before} anyons are relayed.
However,
the practicality of this depends on the hardware implementation
so in this paper we assume the worse case
and relay anyons in the direction of
pointers from the \emph{previous} timestep.

Stage \verb|presyncing| determines which clusters are active.
Note in \cref{fig:surface_code}(a)
boundary nodes are of lower \verb|ID| than detectors.
This ensures that if a cluster touches a boundary,
its root is a boundary node.
Each root can thus determine if its cluster is active
by checking
	if itself is a boundary node and
	if it has an anyon,
by \cref{lem:cluster_activity}.
In \verb|presyncing|,
roots of active clusters set their \verb|active| attribute to \verb|true|;
every other node sets theirs to \verb|false|.

\Cref{fig:almost_local_syncing} shows a \verb|syncing| example.
\begin{figure}
	\centering
	\includegraphics[width=0.5\textwidth]{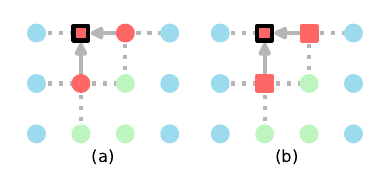}
	\caption{An example of a \texttt{syncing} stage which needs two timesteps.
	Attributes are shown as in \cref{fig:almost_local_merging}
	but without \texttt{CID}.
	(a) A \texttt{presyncing} stage has just finished.
	Note how the only active node is the root of the three-node active cluster.
	(b) Nodes 7 and 8 see this active node along a fully grown edge
	so become active themselves.
	This is the end of the flood
	as all the nodes in the active cluster are active.}
	\label{fig:almost_local_syncing}
\end{figure}
\begin{figure*}
\begin{minipage}{\linewidth}
\begin{algorithm}[H]
\caption{Run by each node during burning.}
\label{alg:burning}
\begin{algorithmic}
	\Procedure{Burning}{$v$}
		\For{$uv \in v.\texttt{owned}$}
			\If{$uv.\texttt{growth} = 1$ and not
			($u.\texttt{pointer}$ toward $v$ or $v.\texttt{pointer}$ toward $u$)}
				\State $uv.\texttt{growth} \gets 0$
			\EndIf
		\EndFor
	\EndProcedure
\end{algorithmic}
\end{algorithm}
\end{minipage}
\end{figure*}
\begin{figure*}
	\centering
	\includegraphics[width=\textwidth]{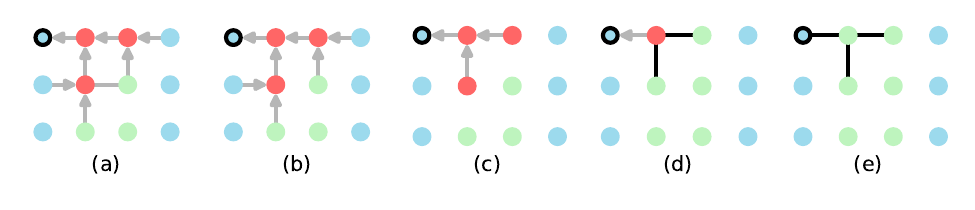}
	\caption{An example of burning and peeling.
	Attributes are shown as in \cref{fig:almost_local_syncing}.
	(a) Before burning,
	a cluster contains a cycle,
	which has one edge no pointer points along.
	(b) This edge is burned i.e.\ removed to leave behind a spanning forest.
	(c) The four leaves of the forest are peeled
	resulting in a smaller forest.
	(d) The two leaves of the smaller forest are peeled
	resulting in an even smaller forest.
	These two edges are added to the correction
	because the leaf nodes were defects.
	Edges so far added to the correction are in black.
	(e) The final leaf is peeled and added to the correction.}
	\label{fig:almost_local_burning_peeling}
\end{figure*}
In \verb|syncing|
the activity of each cluster is propagated
from root to all its other nodes via a flood:
each node looks at its access;
if it sees an active node,
it becomes active,
if not already.
At the end of this stage,
a node is in an active cluster iff it is active.

In Helios' messaging system,
messages contain the \verb|ID| of the destination node.
The memory needed to store this scales logarithmically in $d$.
A new message is emitted whenever a defect changes its \verb|CID|
which leads to $\mathcal O(N^2)$ messages throughout syndrome validation
in the worst case.
Messages must be kept separate
so buffers are needed
but these can stall if full.
Anyon annihilation improves upon this system as
\verb|anyon| is one bit
regardless of $d$.
Throughout syndrome validation,
the anyon count never exceeds the defect count in $G$
and only ever decreases,
so is $\mathcal O(N)$ in the worst case.
There is no need for buffers
(as anyons need not be kept separate)
nor additional links between non-neighbouring nodes.

\subsubsection{Burning}
\label{sec:burning}
In original UF a spanning forest is grown in $\mathcal O(N)$ timesteps
via breadth- or depth-first search.
Macar does this
in one timestep by retrieving the forest directly from the \verb|pointer|s
in a stage we call \emph{burning},
described in \cref{alg:burning}.
\Cref{fig:almost_local_burning_peeling}(a, b) shows an example.
In burning,
any fully grown edge no pointer points along
is `burned' i.e.\ sets its \verb|growth| to 0.
The spanning forest is the set of fully grown edges which remain.
This is because the pointers already define a spanning tree
within each cluster;
burning simply removes the edges `unused' by these pointers.

\subsubsection{Peeling}
\label{sec:peeling2}
\Cref{fig:almost_local_burning_peeling}(b--e) shows
how peeling is done in Macar;
this is simply a local version of that described in \cref{sec:peeling}
and takes $\mathcal O(N)$ timesteps.
In a timestep each node runs the procedure in \cref{alg:peeling}:
each node checks if it is a leaf node of the spanning forest;
if so it will peel its corresponding leaf edge.
\begin{algorithm}[H]
\caption{Run by each node during peeling.}\label{alg:peeling}
\begin{algorithmic}[2]
	\Procedure{Peeling}{$v$}
		\State $v.\texttt{busy} \gets \texttt{false}$
		\If{($v$ not a root) and $`\big(|{\operatorname{access} v}| = 1)$}
			\State $v.\texttt{busy} \gets \texttt{true}$
			\State $`{u} := \operatorname{access} v$
			\State $uv.\texttt{growth} \gets 0$ \label{line:growth_gets_0}
			\If{$v.\texttt{defect}$}
				\State $v.\texttt{defect} \gets \texttt{false}$
				\State add $uv$ to $\mathbb C$ \label{line:add_uv_to_C}
				\State flip $u.\texttt{defect}$
			\EndIf
		\EndIf
	\EndProcedure
\end{algorithmic}
\end{algorithm}

\subsubsection{Runtime}
Macar's worst-case decoding runtime is $\mathcal O(N)$
due to the same reason as Helios':
flooding a cluster of diameter $\mathcal O(N)$
is the process with the worst scaling
in the whole decoding cycle.
Specifically,
the maximum number of timesteps to flood a cluster
equals its diameter.
\Cref{sec:runtime_analysis} shows
the mean runtime scales similarly to Helios'
as expected.

\section{Actis}
\label{sec:actis}
In this section we focus on circuit-level noise
and present a strictly local UF
i.e.\ one without a direct communication link between the controller and each node.
Instead,
the controller communicates only with the node of \verb|ID| 0.
Broadcasting of \verb|stage| from controller to nodes
happens via a global flood,
and \verb|busy| and \verb|active| propagate as signals
from node to controller by relaying.

These two new processes,
which we call
\emph{staging}
and \emph{signalling}
respectively,
are communicated through the local tree $T_\text{SS}$
shown in \cref{fig:SS_tree}(b).
\begin{figure}
	\centering
	\includegraphics[width=0.5\textwidth]{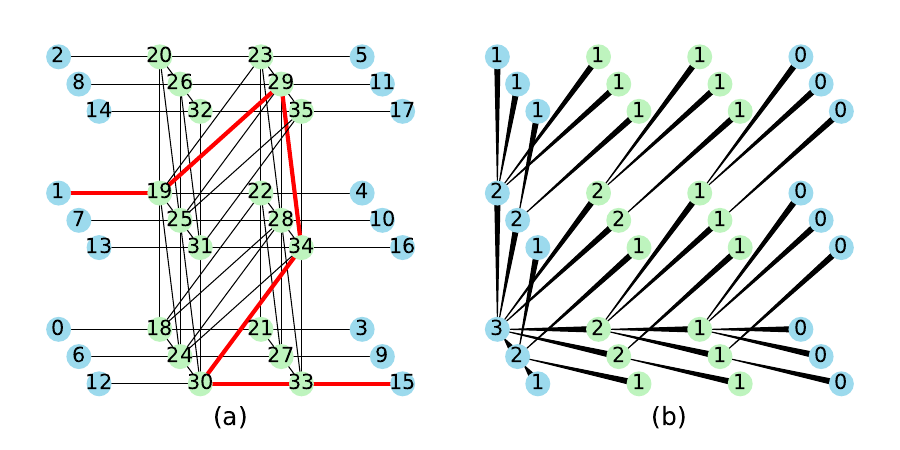}
	\caption{(a) The generalisation of \cref{fig:3D_surface_code}(a)
	to circuit-level noise.
	The diagonal edges exist in every cubelet in the bulk of $G$.
	The thick red path is an example of a logical bitflip.
	(b) The tree $T_\text{SS}$ used in Actis,
	with nodes labelled by their \texttt{span}.
	Each edge tapers toward the parent.
	}
	\label{fig:SS_tree}
\end{figure}
Signalling is similar to how anyons are relayed
but instead of following \emph{variable} pointers,
these signals follow \emph{constant} pointers
toward the controller,
traversing one edge in $T_\text{SS}$ per timestep.
This requires the controller and nodes
to store additional attributes
which are described in \cref{sec:additional_attributes}.

The change clearly increases runtime
but has the advantage of lacking long-range links,
so is an even lower stepping stone toward practicality.
The worst-case runtime
	remains $\mathcal O(N)$
	and is discussed in \cref{sec:runtime_2}.
We will see in \cref{sec:runtime_analysis}
the mean runtime is below
$\mathcal O(d^2)$ for all practical noise levels.
Moreover we note
that a rapid-transmission variant
can substantially recover the speed of Macar
within the strictly local paradigm.

\subsection{The Staging and Signalling Tree}
\label{sec:the_staging_and_signalling_tree}
As \cref{fig:SS_tree}(b) shows,
the root of $T_\text{SS}$ is node 0.
Its height
i.e.\ distance from root to a furthest node,
equals the code distance: $h(T_\text{SS}) =d$.
Almost all edges in $T_\text{SS}$ are diagonals
which already exist in $G$ for circuit-level noise.
The only additional edges are in the
	west,
	east,
	and lower
faces of $G$,
whose count is $\mathcal O(d^2)$
compared to $\mathcal O(d^3)$ for edges in $G$.
Therefore the physical overhead of $T_\text{SS}$ is low.

$T_\text{SS}$ differs from the
hierarchical tree of Helios' controller:
$T_\text{SS}$ is mostly embedded in $G$
and is truly local,
unlike the latter which is separate from $G$
and intersects $V$ only at its leaves.

\subsection{Additional Attributes}
\label{sec:additional_attributes}
The controller stores a constant integer
$\texttt{span} =1 +h(T_\text{SS})$
and has the following variable attributes:
\begin{itemize}
	\item \verb|countdown| is an integer
	which tracks how long the controller must wait
	until it is sure staging or signalling is done.
	It is initialised to 0.
	\item \verb|busy_signal| is a boolean indicating
	whether the controller receives a busy signal
	in the current timestep.
	\item \verb|active_signal| is a boolean indicating
	whether the controller \emph{has} received an active signal
	during the current \verb|syncing| \emph{stage}.
	It is initialised to \verb|false|.
\end{itemize}
Each node $v \in V$ stores the following constants:
\begin{itemize}
	\item $\texttt{span} =h(T_\text{SS}) -\operatorname{depth}v$,
	where $\operatorname{depth}v$ is the distance in $T_\text{SS}$
	from $v$ the root.
	\item \texttt{signalee} is the parent of $v$ in $T_\text{SS}$
	i.e.\ the neighbour in the direction toward the controller.
	This is the neighbour that $v$ looks to for staging,
	and sends busy and active signals to.
	For node 0, \verb|signalee| is the controller.
\end{itemize}
Each node $v$ also has the following variable attributes:
\begin{itemize}
	\item \verb|stage| is used for staging;
	possible values are the seven stages.
	$v.\verb|stage|$ is initialised to \verb|growing|.
	\item \verb|countdown| is same as that of controller
	but used only to track how long $v$ must wait until
	staging is done.
	$v.\verb|countdown|$ is initialised to 0.
	\item \verb|busy_signal| is a boolean
	indicating if in the current timestep,
	$v$ either
		is busy
		or receives a busy signal from another node.
	\item \verb|active_signal| is a boolean
	indicating if in the current timestep,
	$v$ either
		has \verb|stage| as \verb|presyncing| and becomes active,
		or receives an active signal from another node.
\end{itemize}
To manage these additional attributes
the procedures have additional steps,
explained in the next subsection.

\begin{algorithm}[H]
\caption{Run by controller every timestep in Actis.}
\label{alg:controller_advance_actis}
\begin{algorithmic}[2]
	\Procedure{AdvanceLocalController}{}
		\If{$\texttt{countdown} = 0$}
			\LComment{\texttt{stage} change block.}
			\If{$\texttt{stage} = \texttt{syncing}$}
				\If{not \texttt{active\_signal}}
					\State $\texttt{stage} \gets \texttt{burning}$
				\Else
					\State $\texttt{stage} \gets \texttt{growing}$
				\EndIf
			\Else
				\State $\texttt{stage} \gets$ next stage
			\EndIf
			\LComment{\texttt{countdown} reset block.}
			\If{$\texttt{stage} \in `{
				\texttt{merging},
				\texttt{peeling}
			}$}
				\State $\texttt{countdown} \gets \texttt{span} + 1$
				\label{line:countdown_gets_span+1}
			\Else
				\State $\texttt{countdown} \gets \texttt{span}$
				\label{line:countdown_gets_span}
			\EndIf
		\ElsIf{\texttt{busy\_signal} and $\texttt{countdown} \le 2$}
			\State $\texttt{countdown} \gets 2$
			\label{line:doppler}
			\Comment{Only occurs during arbitrary-duration stages.}
		\Else
			\State $\texttt{countdown} \gets \texttt{countdown} -1$
		\EndIf
	\EndProcedure
\end{algorithmic}
\end{algorithm}

\begin{figure*}
\begin{minipage}{\textwidth}
\begin{algorithm}[H]
\caption{Each node in Actis
runs \textsc{AdvanceLocalNode} every timestep.}
\label{alg:wrappers}
\begin{algorithmic}[2]
	\Procedure{AdvanceLocalNode}{$v$}
		\State $\Call{Proc}{} \gets$ procedure matching $v.\texttt{stage}$
		in \cref{alg:node_advance,alg:burning,alg:peeling}
		\If{$v.\texttt{stage} \in `{
			\texttt{merging},
			\texttt{syncing},
			\texttt{peeling}
		}$}
			\State \Call{AdvanceArbitrary}{$v$, \Call{Proc}{}}
		\Else
			\State \Call{AdvanceFixed}{$v$, \Call{Proc}{}}
		\EndIf
	\EndProcedure
	\State
	\Procedure{AdvanceFixed}{$v$, \Call{Proc}{}}
	\Comment{For fixed-duration stages.}
		\If{$v.\texttt{countdown} = 0$}
		\label{line:countdown_reaches_0}
			\State \Call{Proc}{$v$}
			\If{$v.\texttt{stage} = \texttt{presyncing}$}
				\State $v.\texttt{active\_signal} \gets v.\texttt{active}$
			\EndIf
			\State $v.\texttt{stage} \gets$ next stage
		\Else
			\State $\texttt{countdown} \gets \texttt{countdown} - 1$
		\EndIf
	\EndProcedure
	\State
	\Procedure{AdvanceArbitrary}{$v$, \Call{Proc}{}}
	\Comment{For arbitrary-duration stages.}
		\If{$v.\texttt{stage} = v.\texttt{signalee}.\texttt{stage}$}
			\State \Call{Proc}{$v$}
			\State $v.\texttt{busy\_signal} \gets v.\texttt{busy\_signal}$ or $v.\texttt{busy}$
			\State relay $v.\texttt{busy\_signal}$ and $v.\texttt{active\_signal}$ to $v.\texttt{signalee}$
		\Else
			\State $v.\texttt{stage} \gets v.\texttt{signalee}.\texttt{stage}$
			\label{line:stage_flooding}
			\Comment{Staging.}
			\State $v.\texttt{countdown} \gets v.\texttt{span}$
		\EndIf
	\EndProcedure
\end{algorithmic}
\end{algorithm}
\end{minipage}
\end{figure*}
\subsection{Modified Procedures}\label{sec:modified_procedures}
For the controller,
\cref{alg:controller_advance} is replaced by
\cref{alg:controller_advance_actis}.
Main differences are:
	checking for any busy nodes is replaced
		by checking if \verb|countdown| has reached 0;
	checking for any active nodes,
		by checking if an active signal has been received.

For Macar
the four stages that lasted one timestep
each followed an arbitrary-duration stage,
so in Actis,
staging must be used to broadcast
these four stages to all nodes.
As indicated in \Cref{line:countdown_gets_span},
staging lasts \verb|span| timesteps:
the duration for a \verb|stage| flood to propagate
from controller to the furthest nodes.
Hence these four stages are still of fixed duration
(for fixed $d$).

The other three (arbitrary-duration) stages need not be broadcast
as they each follow a fixed-duration stage
so all nodes know exactly when to start them.
In \verb|merging| and \verb|peeling|
the controller must wait at least $\verb|span|+1$ timesteps
as it takes \verb|span| timesteps for a signal to go
	from a furthest node
	to the controller.
\Cref{line:countdown_gets_span+1} indicates this.
In \verb|syncing| the furthest nodes
are never busy nor active (all are boundary nodes)
so the controller need only wait at least \verb|span| timesteps
(\cref{line:countdown_gets_span}).
In any of these three stages,
after this minimum duration,
busy signals are no more than two edges apart
due to \cref{prop:doppler}
so \verb|countdown| resets to 2 in \cref{line:doppler}.
\begin{prop}[Doppler effect]\label{prop:doppler}
For any busy signal received by the controller
after $\verb|span| +1$ timesteps in \verb|merging|,
another busy signal must have been received
\num{\le 2} timesteps earlier.
\end{prop}
\begin{proof}
All busy signals emitted in the first \verb|merging| timestep
reach the controller in $\verb|span| +1$ timesteps or less.
Hence
any busy signal $s$ received after this time
must have been emitted
by some node $v$
at some time $t_s$ \emph{after} the first \verb|merging| timestep.
This implies
at $t_s -1$
some node
$u \in `{v} \cup \operatorname{access} v$
was busy,
so emitted a busy signal $s'$.
The duration between the controller receiving $s'$ and $s$ is maximised
when $u$ is 1 edge closer to the controller than $v$.
In this case $s'$ arrives 2 timesteps before $s$.
\end{proof}
The proof of \cref{prop:doppler} can be understood as a Doppler effect:
waves of busyness travel through clusters
at a speed of one edge per timestep.
Wavefronts emit busy signals
at a frequency of once per timestep.
If a wave travels away from the controller,
the controller will receive these signals
at a frequency of once per two timesteps.

In a timestep
each node runs \textsc{AdvanceLocalNode} in \cref{alg:wrappers}.
This reuses \cref{alg:node_advance,alg:burning,alg:peeling}
but wraps each procedure in either
	\textsc{AdvanceFixed} or
	\textsc{AdvanceArbitrary}
depending on whether the stage is of fixed or arbitrary duration.

A node's \verb|countdown| is used
only when its \verb|stage| is of fixed duration,
during which it should equal the controller's \verb|countdown|.
It tells the node exactly when to change \verb|stage|:
as soon as it reaches 0 in \cref{line:countdown_reaches_0}.
At other times
i.e.\ after arbitrary-duration stages,
the update of \verb|stage| occurs via staging
(\cref{line:stage_flooding}).

\subsection{Worst-Case Runtime}
\label{sec:runtime_2}
The new processes in Actis,
staging and signalling,
take $\mathcal O(d)$ timesteps
and are used $\mathcal O(1)$ times per growth round.
A decoding cycle needs a maximum of $\mathcal O(d)$ growth rounds
so these processes contribute a maximum $\mathcal O(d^2)$ timesteps
to the overall runtime.
Hence the worst-case runtime scaling of Actis remains $\mathcal O(N)$,
where $N =\mathcal O(d^3)$ for circuit-level noise.

\section{Runtime Analysis}\label{sec:runtime_analysis}
\begin{figure}
	\centering
	\includegraphics[width=0.5\textwidth]{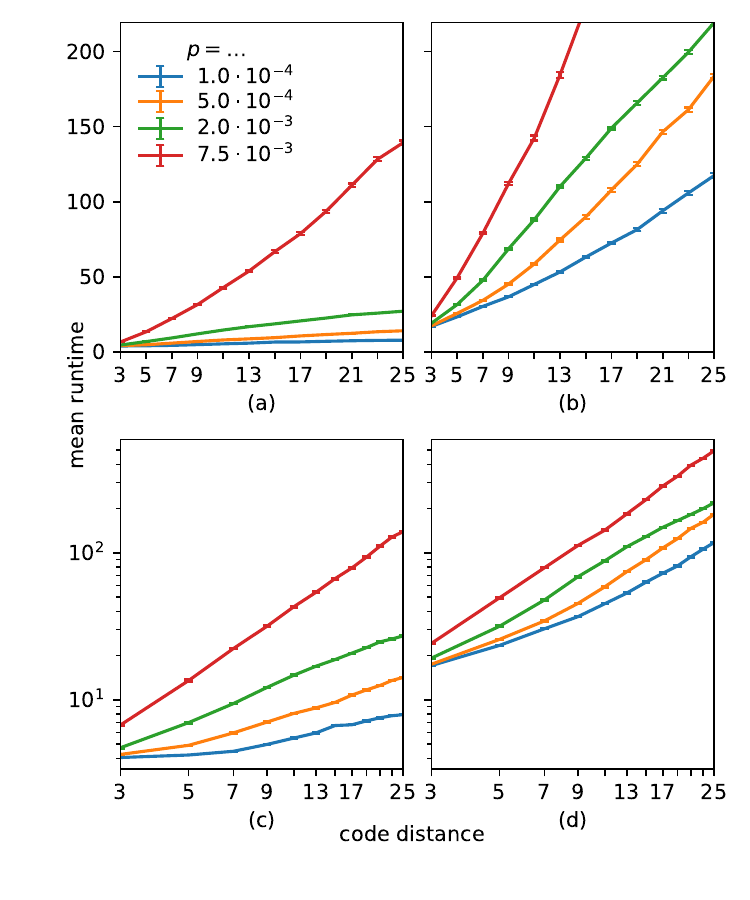}
	\caption{Syndrome validation runtime
	against code distance,
	under circuit-level noise
	for different physical error probabilities $p$.
	Each datapoint is the mean of \num{1e3} samples of decoding cycle;
	errorbars show standard error.
	(a) Macar.
	(b) Actis.
	(c, d) same as (a, b) but on a log--log scale.}
	\label{fig:mean_runtime}
\end{figure}
We wrote a Python package \cite{Chan2023a_quantum_bibstyle} implementing
	original UF,
	Macar
	and Actis.
We then tested the runtimes of Macar and Actis
by emulating a decoding cycle
under circuit-level noise
and counting the number of timesteps to complete syndrome validation.
This is a good proxy for decoding runtime as
we already know the exact runtime of burning,
and peeling is never slower than
the last round of merging.

\begin{rem*}
\verb|burning| and \verb|peeling| need not be standalone stages.
We could have incorporated peeling into merging if,
instead of relaying an anyon along an edge,
we peeled that edge,
with the modifications in \cref{alg:peeling}
that \cref{line:growth_gets_0} is not done
and \cref{line:add_uv_to_C} becomes $\mathbb C \gets \mathbb C \sd `{uv}$.
Defects would then take over the role of anyons
and the decoding cycle would finish after the last \verb|syncing| stage,
without need for burning.
Although small,
this modification may be useful in future implementations
for which we want to minimise the number of stages
e.g.\ a local UF stream decoder.
\end{rem*}

Note the runtimes we record are for correcting only bitflips.
The full decoding cycle in practice requires
two copies of the decoder:
one for each error type (bit- and phaseflip).
Since these copies can operate simultaneously
the overall runtime is always the slower of the two.
The runtimes presented here are thus to be interpreted as raw,
as they demonstrate the scaling of the fundamental algorithm
but are not subject to the above effect.

\Cref{fig:mean_runtime} shows the mean runtime
of Macar and Actis.
We use this data to estimate in \cref{tab:gradients}
the scaling $m$,
assuming mean runtime $\propto d^m$.
\begin{table}
\caption{Gradient $m$ of the lines in \cref{fig:mean_runtime}(c, d)
estimated using Weighted Least Squares.}
\centering
\begin{tabular}{ccc}
	\hline
	\hline
	$p$	& Macar	& Actis \\
	\hline
	$\num{1.0e-4}$	& 0.27(3)	& 0.77(3) \\
	$\num{5.0e-4}$	& 0.55(3)	& 1.04(4) \\
	$\num{2.0e-3}$	& 0.86(1)	& 1.19(2) \\
	$\num{7.5e-3}$	& 1.48(1)	& 1.46(2) \\
	\hline
	\hline
\end{tabular}
\label{tab:gradients}
\end{table}
As \num{\approx 7.5e-3} is the threshold of UF,
\num{2.0e-3} is roughly the highest practical $p$
at which UF would operate.
For this value of $p$ and below:
Macar shows sublinear scaling;
Actis, subquadratic scaling.
We expect the former
as Macar is closely related to Helios
which scales sublinearly \cite{Liyanage2023}.
Note we implement Macar
with a hierarchical tree height of 1
to isolate the behaviour of \cref{alg:node_advance},
though such a tree is not scalable in $d$.
Runtimes for any scalable implementation of Macar
would incur an extra $\mathcal O(\lg d)$ contribution
to those presented here
due to a height-$\mathcal O(\lg N)$ tree.

We also evaluate runtime distribution,
shown in \cref{fig:runtime_distribution}.
\begin{figure*}
	\centering
	\includegraphics[width=\textwidth]{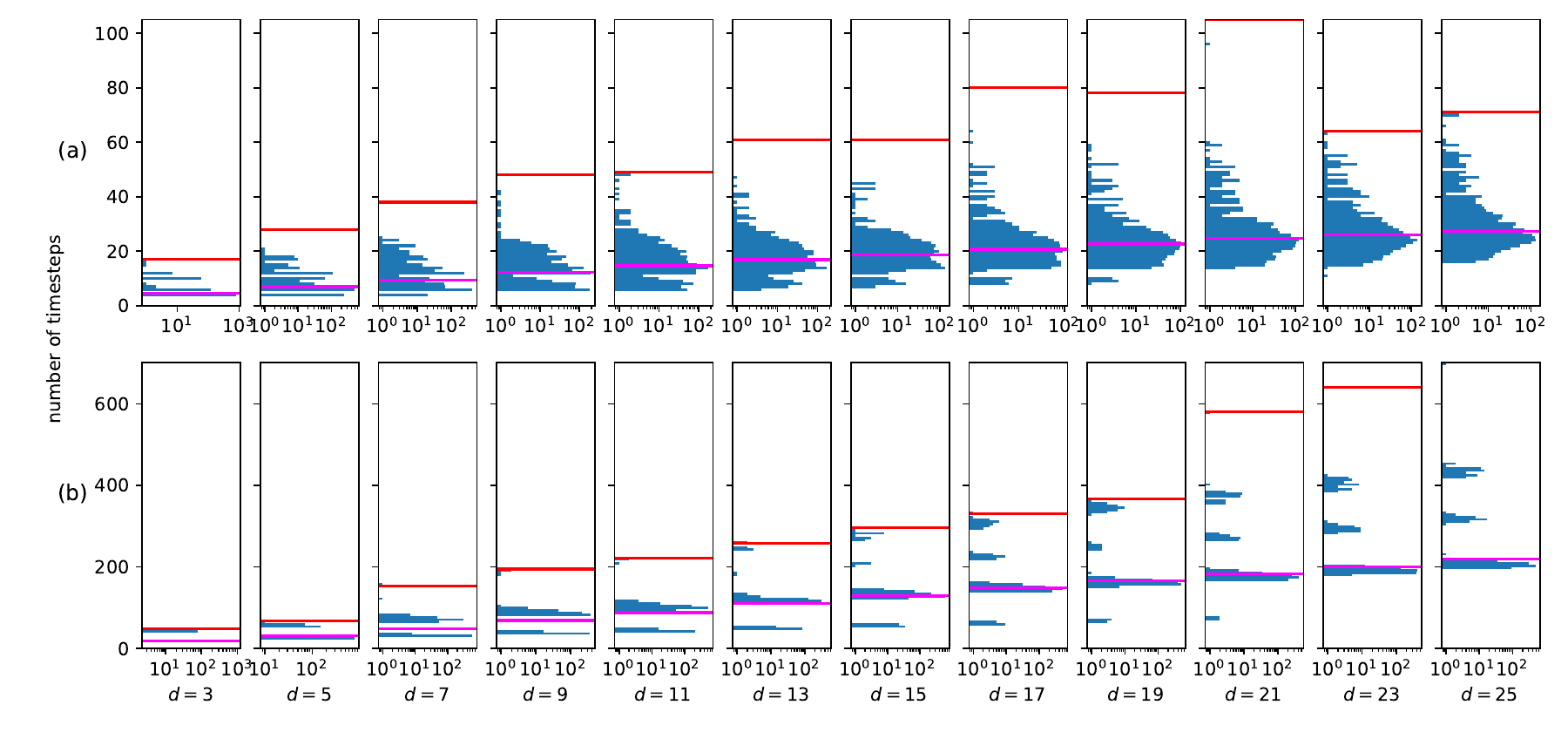}
	\caption{Syndrome validation runtime
	for physical error probability $p =\num{2.0e-3}$.
	Within each histogram:
	the horizontal magenta line shows the mean;
	red line, the maximum,
	of \num{1e4} samples of decoding cycle.
	(a) Macar; each bar has width 1.
	(b) Actis; each bar has width 6.}
	\label{fig:runtime_distribution}
\end{figure*}
Actis has equally spaced sharp peaks in runtime;
each one corresponds to one growth round.
The wide spacing implies most of its runtime is due to
staging and signalling
i.e.\ the controller waiting until countdown reaches 0 before moving on
[this can also be inferred by the difference in absolute runtime
between Macar and Actis:
at $d =25$ the former is 8.06(9) times faster than the latter].

The peak locations can be derived as follows.
If syndrome validation comprises 0 or 1 growth rounds,
\verb|countdown| resets to \verb|span| or $\verb|span|+1$ three times
(see \cref{fig:stage_flowchart}).
For 2 growth rounds,
\verb|countdown| resets seven times.
For $r \ge 0$ growth rounds,
\verb|countdown| resets $\max`{3, 4r-1}$ times.
Hence the peak corresponding to $r$ growth rounds is roughly located at
$\verb|span| \cdot \max`{3, 4r-1}$ timesteps.

\section{Asynchronous Logic}\label{sec:asynchronous_logic}
As just mentioned,
most of the runtime of Actis is due to
staging and signalling.
The fraction of runtime spent on these operations
for the realistic scenario of
	$p =\num{2.0e-3}$
	and $d =25$,
is 0.88(1).
This motivates the question:
can we recover the speed of Macar
whilst maintaining strict locality?

So far we have only considered algorithms
in which each node changes the state of its registers once per timestep,
based on states in previous timesteps.
Signals between nodes therefore
are limited to travel at one edge per timestep.
A timestep would correspond to \emph{at least} one clock cycle of
an FPGA or ASIC
(application-specific integrated circuit)
implementation
of the algorithms described in this paper.

However,
it is also possible for a
register to change state \emph{asynchronously}
i.e.\ as soon as it receives the information to do so
without having to wait for a timestep to complete.
This has one clear advantage,
and a secondary one.
Primarily,
	it is faster:
	signal speed is limited only by
	how fast each register responds to its inputs.
This speed advantage is greatest if used for
	simple,
	indiscriminate
operations like staging and signalling.
The secondary benefit
is that the logic overhead is likely lower
as the number of registered,
i.e.\ clocked,
signals is fewer and consequently
less heat is generated.

Therefore,
it makes sense to consider an Actis variant which uses
	unregistered signals relayed via buffers and/or combinational logic
		for staging and signalling,
	and traditional synchronous logic
		(based on the algorithm timesteps)
		for all other operations.
In the case that it takes less than one synchronous clock cycle
for a signal to travel across the whole array via this technique,
the runtime will be identical to Macar
whilst maintaining strict locality.
We stress that staging and signalling
are ideally suited to
implementation with unregistered signals spanning the array
since there is no risk of race conditions leading to indeterminate behaviour.

\section{Conclusion}\label{sec:conclusion}
We explore UF realised by local architectures,
following recent developments
including the Helios implementation \cite{Liyanage2023}.
We present a streamlined `almost-local' model,
and moreover an extremely practical implementation, Actis,
that is strictly local;
the first of its kind.
We numerically reproduce the sublinear in $d$ mean runtime scaling of Helios
then observe a subquadratic scaling for Actis;
both are better than the
higher-than-cubic scaling of original UF.
Further,
we note Actis is compatible with the use of asynchronous logic
to massively speed up runtime
without sacrificing locality,
making it an attractive design option.

\begin{acknowledgments}
We thank
	Armands Strikis,
	David Garner
	and Nicolas Delfosse
for useful discussions.
The authors would like to acknowledge the use of
the University of Oxford Advanced Research Computing
(ARC)
facility~\cite{Richards2015_quantum_bibstyle} in carrying out this work
and specifically the facilities made available
from the EPSRC QCS Hub grant
(agreement No.\ EP/T001062/1).
The authors also acknowledge support from
two EPSRC projects:
RoaRQ (EP/W032635/1)
and SEEQA (EP/Y004655/1).
TC acknowledges an EPSRC DTP studentship.
Open-source Python libraries used in this work include
\verb|matplotlib|,
\verb|networkx|,
\verb|numpy|,
\verb|pandas|,
\verb|pymatching|,
\verb|pytest|,
\verb|scipy|,
\verb|statsmodels|.
\end{acknowledgments}

\section{Author Contributions}
\label{sec:contributions}
SCB contributed to the core concepts for Macar and Actis,
and authored elements of the paper.
TC established the core concepts Macar and Actis,
implemented,
tested
and validated these concepts,
wrote the emulation code,
collected and processed the numerical data,
and was the primary author of the paper.

\paragraph{Note added:}
During the preparation and initial review process of this paper,
several relevant preprints were updated or announced.
Reference \cite{Liyanage2023}
was updated
to incorporate a concept comparable to our pointer mechanism.
Reference \cite{Griffiths2023} showed that
the Union--Find data structure used in original UF
(but not in our versions)
is redundant and worsens its runtime.
Meanwhile the state of the art
in realising cluster-type decoder hardware
has advanced with a new report~\cite{Barber2023}
describing FPGA and ASIC devices.

\bibliographystyle{quantum}
\bibliography{tchbib}

\appendix
\section{Lemma Proofs}
\subsection{Proof of \cref{lem:leftover}}
\label{sec:proof_of_lem:leftover}
First consider the following.
\begin{lem}[Distributivity of $\sigma$]
The syndrome of the symmetric difference of errors
is the symmetric difference of their syndromes
i.e.\ for $n \ge 2$:
\begin{equation}\label{eq:syndrome_of_n_errors}
\sigma`\bigg(\bigsd_{i=1}^n \mathbb E_i)
=\bigsd_{i=1}^n \sigma(\mathbb E_i).
\end{equation}
\end{lem}
\begin{proof}
By induction on $n$. Consider base case $n =2$:
\begin{equation}
\sigma(\mathbb E_1 \sd \mathbb E_2)
\overset{\eqref{eq:syndrome_definition}}=
V_\d \cap \bigsd_{e \in \mathbb E_1 \sd \mathbb E_2} e.
\end{equation}
Since $e \sd e \equiv\varnothing$ we have
$\bigsd_{e \in \mathbb E_1 \sd \mathbb E_2} e \equiv
(\bigsd_{e \in \mathbb E_1} e) \sd
(\bigsd_{e \in \mathbb E_2} e)$ so
\begin{equation}\label{eq:syndrome_of_2_errors}
\sigma(\mathbb E_1 \sd \mathbb E_2)
\overset{\eqref{eq:syndrome_definition}}=
\sigma(\mathbb E_1) \sd
\sigma(\mathbb E_2)
\end{equation}
so true for $n =2$.
Now for the inductive step;
assume true for $n =k$ and consider for $n =k+1$:
\begin{align}
	\bigsd_{i=1}^{k+1} \mathbb E_i
&=`\bigg(\bigsd_{i=1}^k \mathbb E_i) \sd \mathbb E_{k+1} \notag \\
\sigma`\bigg(\bigsd_{i=1}^{k+1} \mathbb E_i)
&\overset{\eqref{eq:syndrome_of_2_errors}}=
\sigma`\bigg(\bigsd_{i=1}^k \mathbb E_i) \sd
\sigma(\mathbb E_{k+1}) \notag \\
&\overset{\eqref{eq:syndrome_of_n_errors}}=
\bigsd_{i=1}^{k+1} \sigma(\mathbb E_i).
\end{align}
Hence, true for $n =k$ implies true for $n =k+1$.
Together, the base case and inductive step imply true for $n \ge 2$.
\end{proof}
\noindent
Now we can prove \cref{lem:leftover}.
\begin{proof}[Proof of \cref{lem:leftover}]
Consider the syndrome produced by the leftover:
\begin{align}
\sigma(\mathbb L)
&\overset{\eqref{eq:syndrome_of_2_errors}}=
\sigma(\mathbb E) \sd
\sigma(\mathbb C) \notag \\
&=\mathbb{S \sd S} \notag \\
&=\varnothing.
\end{align}
Hence, every detector is incident to an even number of edges in $\mathbb L$.
This can only occur if $\mathbb L$ comprises
cycles or paths whose endpoints are both boundary nodes.
Since each boundary node is incident to exactly one edge
and $\mathbb L$ comprises distinct edges,
these endpoints must be distinct.
\end{proof}

\subsection{Proof of \cref{lem:cluster_activity}}
\label{sec:proof_of_lem:cluster_activity}
We prove the following equivalent statement.
\begin{lem}
Cluster $C$ is inactive iff its defect count is even
or touches a boundary i.e.
\begin{equation}
`\big(\textnormal{$|\mathbb S \cap V_C|$ even})
\vee
`\big(V_\textnormal{boundary} \cap V_C \ne\varnothing).
\end{equation}
\end{lem}
\begin{proof}~
\begin{itemize}
	\item[$\Rightarrow$]
	We know
	$\exists \mathbb C \subseteq E_C :\sigma(\mathbb C) =\mathbb S \cap V_C$.
	If $C$ touches\dots
	\begin{itemize}
		\item no boundary,
		$|\sigma(\mathbb C)|$ even $\forall \mathbb C \subseteq E_C$
		as adding any edge in $E_C$ to $\mathbb C$
		changes $|\sigma(\mathbb C)|$ by either 0 or 2.
		\item a boundary,
		$|\sigma(\mathbb C)|$ can be even or odd
		as adding an edge containing a boundary node to $\mathbb C$
		changes $|\sigma(\mathbb C)|$ by 1.
	\end{itemize}
	
	\item[$\Leftarrow$]
	If $|\mathbb S \cap V_C|$ even,
	partition $\mathbb S \cap V_C$ into $s/2$ pairs
	where $s :=|\mathbb S \cap V_C|$.
	$C$ connected so there exists a path $P_i \subseteq E_C$
	between the defects in each pair.
	Be $\mathbb C :=\bigsd_{i=1}^{s/2} P_i$
	the symmetric difference of these paths then
	\begin{equation}\label{eq:syndrome_of_sd_of_paths}
	\sigma(\mathbb C)
	\overset{\eqref{eq:syndrome_of_n_errors}}=
	\bigsd_{i=1}^{s/2} \sigma(P_i)
	\end{equation}
	where $\sigma(P_i)$ is precisely the $i^\th$ pair so
	\begin{align}
	\sigma(\mathbb C)
	&=\bigcup_{i=1}^{s/2} \sigma(P_i) \notag \\
	&=\mathbb S \cap V_C. \label{eq:inactive_def}
	\end{align}
	The existence of a correction satisfying \cref{eq:inactive_def}
	is the definition for $C$ to be inactive.

	If $|\mathbb S \cap V_C|$ odd,
	$C$ must touch a boundary
	so pick any boundary node $v \in V_C$.
	Partition $(\mathbb S \cap V_C) \cup `{v}$ into $s/2$ pairs
	where $s :=|(\mathbb S \cap V_C) \cup `{v}|$.
	Once again
	there exists a path $P_i$ between the nodes in each pair
	so construct $\mathbb C$ as before
	then \cref{eq:syndrome_of_sd_of_paths,eq:inactive_def} follow.
	\qedhere
\end{itemize}
\end{proof}

\section{Circuit-Level Noise}
\label{sec:circuit-level_noise_appendix}
Depolarising noise is emulated via the following error processes --
with probability:
\begin{itemize}
	\item $p_\M$, initialisation (measurement) prepares (records)
	the opposite state;
	\item $p_1$, gates $\hat\Ident, \hat H$ are followed by
	an error drawn randomly from
	$`{\hat X, \hat Y :=\hat X \hat Z, \hat Z}$;
	\item $p_2$, $\C(\hat X)$ is followed by
	an error drawn randomly from
	$`{\hat\Ident, \hat X, \hat Y, \hat Z}^{\otimes 2}
	\setminus `{\hat\Ident^{\otimes 2}}$.
\end{itemize}
In this paper
we call each of these processes a \emph{fault}.
We apply this emulation
to the syndrome extraction circuits in \cref{fig:5_step}
\begin{figure}
	\centering
	\begin{tabular}{cc}
		\begin{tikzpicture}
			\begin{yquant}
			qubit {\textsc{n}} d[1];
			qubit {\textsc{w}} d[+1];
			qubit {\textsc{e}} d[+1];
			qubit {\textsc{s}} d[+1];
			qubit {\textsc{a}} a;
			cnot a | d[0];
			cnot a | d[1];
			cnot a | d[2];
			cnot a | d[3];
			[type=qubit] measure {\textsc{z}} a;
			\end{yquant}
		\end{tikzpicture}
		&
		\begin{tikzpicture}
			\begin{yquant}
			qubit {\textsc{n}} d[1];
			qubit {\textsc{w}} d[+1];
			qubit {\textsc{e}} d[+1];
			qubit {\textsc{s}} d[+1];
			qubit {\textsc{a}} a;
			cnot d[0] | a;
			cnot d[1] | a;
			cnot d[2] | a;
			cnot d[3] | a;
			[type=qubit] measure {\textsc{x}} a;
			\end{yquant}
		\end{tikzpicture}
		\\
		(a) & (b)
	\end{tabular}
	\caption{The 5-step syndrome extraction circuit for
	(a) correcting bitflips;
	(b) correcting phaseflips.
	\textsc{a} is the ancilla qubit;
	\textsc{n, w, e, s} are the data qubits
	north,
	west,
	east,
	and south of \textsc{a} respectively.
	Each data qubit interacts with a different ancilla qubit
	every step for the first four steps
	but only one interaction is shown here.
	The initial state of \textsc{a} is
	the measurement outcome from the previous cycle
	i.e.\ $\ket{0}$ or $\ket{1}$ for (a),
	and $\ket{+}$ or $\ket{-}$ for (b).}
	\label{fig:5_step}
\end{figure}
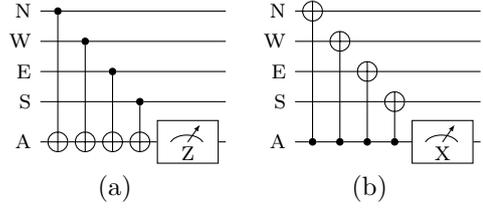
and choose the \emph{balanced parametrisation}:
\begin{equation}\label{eq:balanced_parametrisation}
(p_\M, p_1, p_2) =`\Big(\frac23 \frac45, \frac45, 1) p,
\end{equation}
explained in the next subsection.

\subsection{Parametrisations}\label{sec:parametrisations}
For ease of comparison,
one usually parametrises $(p_\M, p_1, p_2)$
in terms of some characteristic physical error probability $p$
such that $p_\M, p_1, p_2 =\mathcal O(p)$.
There are various such parametrisations
which depend on the hardware used for qubits \cite[p 1]{Wang2011}.
The simplest is the \emph{standard parametrisation}
defined by $p_\M =p_1 =p_2 =p$.
The balanced parametrisation \cref{eq:balanced_parametrisation}
is argued as follows:
\begin{itemize}
	\item
	A qubit involved in a two-qubit gate
	has a probability of error due to it of $12p/15$
	(as 12 of the 15 possible two-qubit gate faults
	result in an error on qubit 1;
	same goes for qubit 2 by symmetry).
	We set this equal to its probability of error
	due to a one-qubit circuit element
	(i.e.\ initialisation, measurement, $\hat\Ident$, or $\hat H$)
	\cite[p 39]{Knill2005}.
	\item Initialisation/measurement is in only one basis
	so is affected by only two errors from $`{\hat X, \hat Y, \hat Z}$.
\end{itemize}
We next describe how circuit-level noise abstracts
into the graph-theoretic approach,
and how to explicitly construct $G$
and the bitflip probability of an edge.

\subsection{Methodology}\label{sec:methodology}
For $G$ in this noise model,
nodes represent the same as they do
in the phenomenological noise model.
Edges however are different:
define bulk (boundary) edges as those
which are (are not) subsets of $V_\d$.
Then each bulk (boundary) edge corresponds to
a possible pair of defects (a possible defect)
that could have resulted from one fault.

To construct $G$
we follow the recipe described in Wang et al.~\cite{Wang2011}.
Namely,
we take the circuits in \cref{fig:5_step}
and tabulate the defect pairs
resulting from each fault for each gate/measurement,
as in \cite[Figure~3]{Wang2011}.
Then,
for each edge orientation,
we gather from the tabulation
all pairs matching that orientation,
as in \cite[Figure~4]{Wang2011}.

Each pair in the tabulation occurs with one of four probabilities from
$\v \pi :=(4p_2/15, 8p_2/15, 2p_1/3, p_\M)^\T$.
Thus for each orientation we define a multiplicity vector
$\v m \in \mathbb Z_{\ge 0}^4$
where $m_i$ is the number of
gathered pairs occurring with probability $\pi_i$.
The total number of pairs gathered by the orientation is $\sum_i m_i$.
The bitflip probability for an edge of this orientation
is the probability of an odd number of pairs occurring:
\begin{align}
\pr(\v m, \v \pi) &=
\sum_{i=1}^4 m_i \pi_i (1 -\pi_i)^{m_i -1}
\prod_{j \ne i} (1 -\pi_j)^{m_j} \notag\\
&+\mathcal O(p^3) \label{eq:bitflip_pr} \\
&\approx `\Big[\prod_{i=1}^4 (1 -\pi_i)^{m_i}]
\sum_{j=1}^4 \frac{m_j \pi_j}{1 -\pi_j}.
\end{align}

We employ special treatment for all four walls of the surface code,
where ancilla (data) qubits idle
instead of interact with a nonexistent data (ancilla) qubit.
Here,
edges are affected by
	more idling faults i.e. faults due to $\hat\Ident$,
	and fewer two-qubit-gate faults
than their counterparts in the bulk.
Multiplicity vectors thus depend not only on edge orientation
but also on their spatial (but not temporal) location in $G$.

\begin{table}
	\caption{Multiplicities for each edge orientation
	when correcting bitflips.
	Each orientation is labelled by
	one of its two corresponding \texttt{pointer} values.
	`On a wall' means the edge is not at a corner but either
		north- or southmost,
		or a boundary edge;
	see \cref{fig:circuit_level_unmerged}.}
	\centering
	\begin{tabular}{rl}
		\hline
		\hline
		Orientation	&	$\v m^\T$ \\
		\hline
		\texttt{S}	&	4210 \\
		\texttt{E}	&	$\begin{cases*}
			1020 &west boundary edge \\
			1010 &east boundary edge \\
			2010 &bulk edge \\
		\end{cases*}$ \\
		\texttt{U}	&	$\begin{cases*}
		3011 &north- or southmost \\
		4001 &else \\
		\end{cases*}$ \\
		\texttt{SD} or \texttt{SEU}	&	2000 \\
		\texttt{EU}	&	$\begin{cases*}
			2010 &at a west corner \\
			2020 &at an east corner \\
			3010 &on a wall \\
			4000 &else \\
		\end{cases*}$ \\
		\hline
		\hline
	\end{tabular}
	\label{tab:multiplicities_x}
\end{table}
\begin{figure}
	\centering
	\includegraphics
	[width=0.3\textwidth]
	{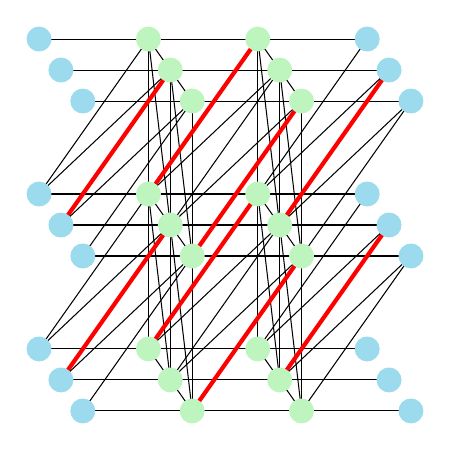}
	\caption{The graph $G_\r$ to correct bitflips
	under circuit-level noise
	with redundant diagonal boundary edges.
	The thick red edges are the set of \texttt{EU} edges `on a wall'.}
	\label{fig:circuit_level_unmerged}
\end{figure}
\Cref{tab:multiplicities_x} shows all multiplicity vectors
(which for brevity we write sans parentheses and commas)
for the graph in \cref{fig:circuit_level_unmerged}.
Note the edge orientations shown
are for correcting bitflips,
so are rotated about the $t$-axis by $\uppi/2$
with respect to those in \cite{Wang2011}
where the graph used corrects phaseflips,
and for which a similar multiplicity table can be computed.

Note also the graph in \cref{fig:circuit_level_unmerged}
has diagonal boundary edges.
This is because,
merely for bookkeeping,
we assign each fault resulting in one defect
the orientation of the resultant defect pair
had the fault occurred in the bulk of the code.
These edges are redundant in the sense that
every detector which needs to connect to the boundary
is already part of a unique horizontal boundary edge.
Mathematically,
an edge $e$ is redundant if there already exists another edge $f$
with the same syndrome:
$\sigma(`{e}) =\sigma(`{f})$.

Error sampling and UF
are unaffected by this redundancy,
but MWPM and weighted-edge variants of UF
only reach their full potential when there are no redundant edges.
Therefore, as a final step we merge all redundant edges in $G_\r$,
resulting in the graph $G$ in \cref{fig:SS_tree}(a).
To merge an edge $e$ with $f$
we delete $e$
and add the multiplicity vector of $e$ onto that of $f$.

\subsection{Other Syndrome Extraction Circuits}
\begin{figure*}
	\centering
	\includegraphics[width=\textwidth]{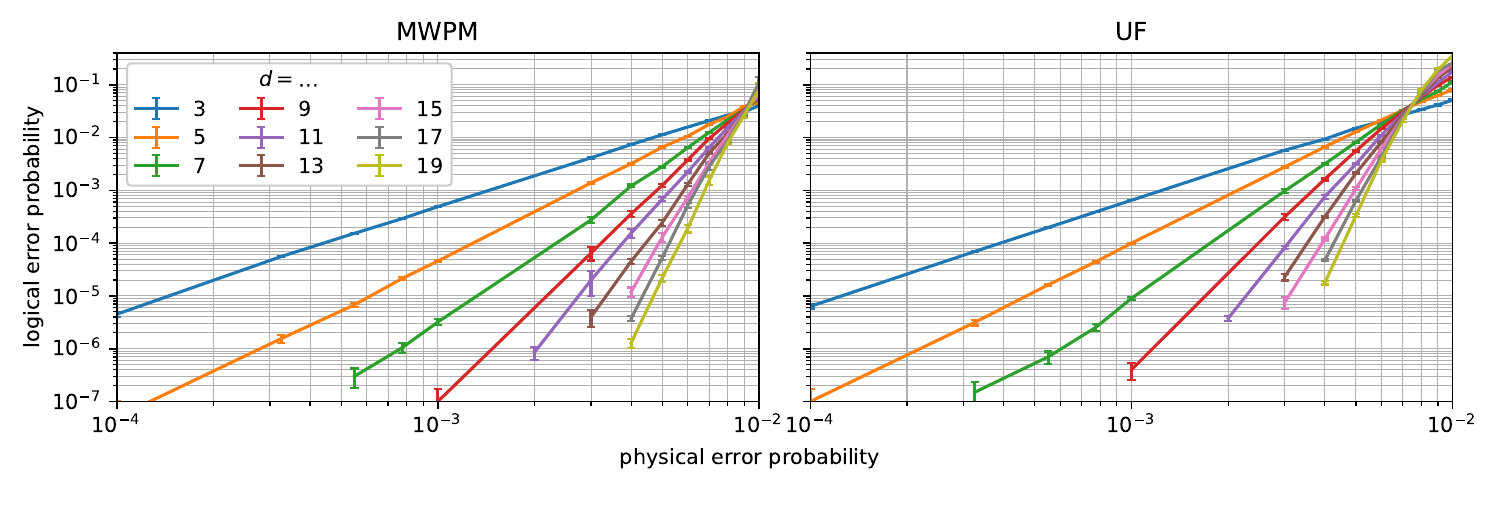}
	\caption{Threshold plots for MWPM and UF
	under the circuit-level noise model described in
	\cref{sec:circuit-level_noise_appendix}.
	Each datapoint is the mean of
	\numrange{1e2}{3e7} samples;
	errorbars show standard error.
	The threshold of MWPM is \num{\approx 9.2e-3}.}
	\label{fig:MWPM_v_UF}
\end{figure*}
The circuits in \cref{fig:5_step}
assume quantum nondemolition measurements
(so initialisation is not used)
which are native in both Z and X bases.
To be more pessimistic about quantum hardware
we can consider the following independent restrictions,
which both increase the step count of the circuit
and affect multiplicity vectors by
$\v m \gets \v m +\upDelta\v m$
where we specify $\upDelta\v m$ below.
\begin{itemize}
	\item Assuming demolition measurements
	separates measurement and initialisation.
	This allows both processes to err independently
	which increments $m_4$ for vertical edges.
	Also,
	all data qubits must idle for an extra step
	which increments $m_3$ for horizontal edges.
	The result is
	\begin{equation}
	\upDelta\v m^\T =\begin{cases*}
	0010 &\texttt{S} or \texttt{E} \\
	0001 &\texttt{U} \\
	0000 &else, \\
	\end{cases*}
	\end{equation}
	and a step count increase of 1.
	\item Restricting initialisation (if used) and measurement to,
	say the Z basis,
	means X-basis initialisation and measurement
	each require a Hadamard.
	This increases the step count of \emph{both} circuits
	(if they are performed with equal frequency)
	by 2,
	hence adds 2 to $m_3$ for horizontal and vertical edges:
	\begin{equation}\label{eq:Delta_m_monolingual}
	\upDelta\v m^\T =\begin{cases*}
	0020 &\texttt{S} or \texttt{E} or \texttt{U} \\
	0000 &else.
	\end{cases*}
	\end{equation}
\end{itemize}
\Cref{fig:8_step} shows syndrome extraction
under the combination of these two restrictions.
\begin{figure}
	\centering
	\begin{tabular}{cc}
		\begin{tikzpicture}
			\begin{yquant}
			qubit {\textsc{n}} d[1];
			qubit {\textsc{w}} d[+1];
			qubit {\textsc{e}} d[+1];
			qubit {\textsc{s}} d[+1];
			hspace {5mm} d;
			[after=d, inner sep=0mm]
			qubit {$\ket{0}$} a;
			cnot a | d[0];
			cnot a | d[1];
			cnot a | d[2];
			cnot a | d[3];
			dmeter {Z} a;
			discard a;
			\end{yquant}
		\end{tikzpicture}
		&
		\begin{tikzpicture}
			\begin{yquant}
			qubit {\textsc{n}} d[1];
			qubit {\textsc{w}} d[+1];
			qubit {\textsc{e}} d[+1];
			qubit {\textsc{s}} d[+1];
			hspace {5mm} d;
			[after=d, inner sep=0mm]
			qubit {$\ket{0}$} a;
			box {H} a;
			cnot d[0] | a;
			cnot d[1] | a;
			cnot d[2] | a;
			cnot d[3] | a;
			box {H} a;
			dmeter {Z} a;
			discard a;
			\end{yquant}
		\end{tikzpicture}
		\\
		(a) & (b)
	\end{tabular}
	\caption{The 8-step syndrome extraction circuit for
	(a) correcting bitflips;
	(b) correcting phaseflips.}
	\label{fig:8_step}
\end{figure}
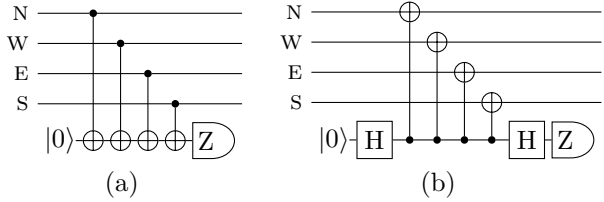
Note we have slightly optimised the circuit in (a)
by initialising and measuring as late and as early as possible,
respectively,
so the ancilla qubit never suffers from idling faults.
Thus for this circuit, \cref{eq:Delta_m_monolingual} is replaced by
\begin{equation}
	\upDelta\v m^\T =\begin{cases*}
0020 &\texttt{S} or \texttt{E} \\
0000 &else.
\end{cases*}
\end{equation}
For small $\v \pi$,
an increase in $\v m$ generally leads to
an increase in $\pr(\v m, \v \pi)$ from \cref{eq:bitflip_pr}
hence a decrease in decoder threshold.
Indeed,
for UF under the noise model
defined by the circuit in (b)
we numerically observe a threshold of \num{\approx 4.5e-3}
(cf.\ \cref{fig:thresholds}).
Generally in this paper,
we use the circuits in \cref{fig:5_step}.

\section{Comparison with MWPM}
\label{sec:comparison_with_mwpm}
\begin{rem*}
While both Helios and our versions of UF
grow clusters identically to original UF,
their \emph{final} outputs
given the same input $\mathbb S$ may not be identical,
as $\mathbb C$ depends on the choice of the root
and spanning tree of each cluster.

It is even possible that,
given an error $\mathbb E$ which generates $\mathbb S$,
one of the corrections will cause a logical error while another will not.
In other words,
the two corrections differ by a logical operation.
However,
such errors are those where the UF approach itself
cannot provide a preference between the two corrections --
the choice between them is arbitrary.
Indeed,
the above scenario necessarily implies
a cluster spans between opposite boundaries.
Since UF arbitrarily pairs defects within a cluster,
any logical difference between corrections
within this cluster is random.

One would therefore not expect these differences to lead to
any observable difference in decoder \emph{accuracy},
at least in terms of the subthreshold gradient.
This is indeed the case,
as \cref{fig:thresholds} shows.
\end{rem*}

With this remark,
we can unambiguously talk about the accuracy of UF
without specifying the particular implementation.
\Cref{fig:MWPM_v_UF} shows the accuracy of MWPM and UF
for practical noise levels and code distances.
We implement MWPM using Sparse Blossom \cite{Higgott2023}
and assign each edge a weight $\ln[(1-p_e)/p_e]$
where $p_e$ is its bitflip probability.
For both decoders below their respective thresholds,
we assume the logical error probability $f$ follows a power law:
\begin{equation}
\frac f{f_\th} =`\Big(\frac p{p_\th})^m
\end{equation}
where $(p, f)_\th$ are the coordinates of the threshold.
\cref{fig:gradients} shows both decoders have
the same subthreshold gradient $m(d)$,
suggesting their accuracy difference
can be characterised solely by their thresholds.
\begin{figure}
	\centering
	\includegraphics[width=0.5\textwidth]{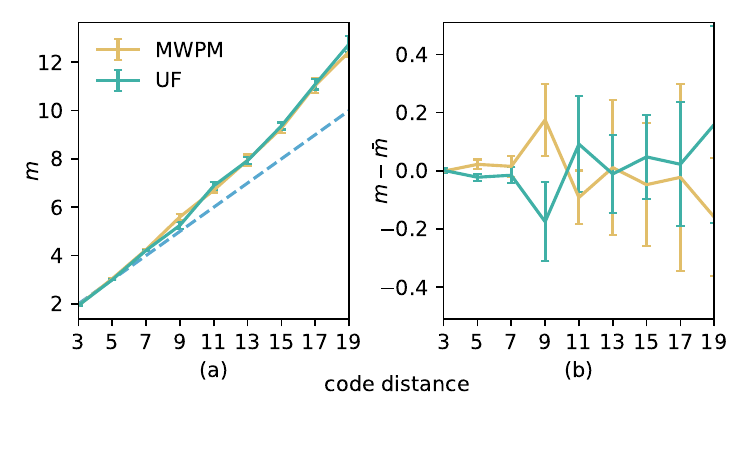}
	\caption{(a) Subthreshold gradient of the lines in
	\cref{fig:MWPM_v_UF}
	estimated using Weighted Least Squares.
	The dashed blue line shows the asymptotic ($p \to 0$) value $m =(d+1)/2$ from theory
	which underestimates the actual value at high distances --
	a feature already documented for MWPM \cite[\S 5]{Fowler2012b}.
	(b) Deviation of said gradient from the mean
	$\bar m :=\frac12(m_\text{MWPM} +m_\text{UF})$
	of both decoders.}
	\label{fig:gradients}
\end{figure}
These are around
	$(0.92, 3.8)\cdot\num{1e-2}$ for MWPM
	and $(0.75, 4.3)\cdot\num{1e-2}$ for UF,
so the former is more accurate.

Recent developments
have implemented MWPM with
	almost-linear mean runtime scaling \cite{Higgott2023}
	and parallelisation \cite{Wu2023}.
One may therefore question the relevance of UF
given it is less accurate.
However,
as UF approximates MWPM \cite{Wu2022},
any hardware capable of implementing MWPM
would likely be able to implement UF
and decode faster in the absolute sense.
Moreover,
and to the best of our knowledge,
no one has implemented MWPM in a strictly local fashion.
These reasons maintain UF as a competitive decoder,
especially in the near term.

\end{document}